\documentclass[12pt]{article}
\usepackage[all]{xy}
\usepackage{amsmath, amssymb, amsfonts, amsthm, mathrsfs, dsfont, mathtools, graphics, stackrel, multirow, rotating, paralist, extarrows}%, hyperref}
\usepackage[english]{babel}
\usepackage[toc,page,header]{appendix}
\DeclareMathAlphabet{\mathpzc}{OT1}{pzc}{m}{it}

\numberwithin{equation}{section}

\renewcommand{\theequation}{\arabic{section}.\arabic{equation}}

\theoremstyle{plain}
\newtheorem{theorem}{Theorem}[section]
\newtheorem{lemma}{Lemma}[section]
\newtheorem{proposition}{Proposition}[section]
\newtheorem{corollary}{Corollary}[section]

\theoremstyle{definition}
\newtheorem{definition}{Definition}[section]

\theoremstyle{remark}
\newtheorem{remark}{\bf Remark}[section]

\newcommand{\oversetc}[1]{\overset{\text{\tiny}{\text{\tiny}\circ}}#1}
\newcommand\undersym[2]{\raisebox{-6pt}{\tiny$#2$}{\kern-5pt}\mbox{$#1$}}
\newcommand\overcirc[1]{\raisebox{10pt}{\tiny$\circ$}{\kern-7pt}\mbox{$#1$}}
\title{ \bf \Huge{A Global Approach to Absolute \\ Parallelism Geometry}\footnote{arXiv: 1209.1379 [gr-qc]}}

\author{Nabil L. Youssef $^{1,\,2}$ and Waleed A. Elsayed $^{1,\,2}$}

\date{ }

\begin{document}

\bibliographystyle{plain}
\maketitle                     % Produces the title.
\vspace*{-.95cm}

\begin{center}
${^1}$Department of Mathematics, Faculty of Science, \\
Cairo University, Giza, Egypt \\
\vspace*{.5cm}
${^2}$Center for Theoretical Physics (CTP),\\
British University in Egypt (BUE)
\end{center}

\begin{center}
$
\begin{array}{ll}
\text{Emails:}&\text{nlyoussef@sci.cu.edu.eg, nlyoussef2003@yahoo.fr} \\
  &\text{waleedelsayed@sci.cu.edu.eg, waleed.a.elsayed@gmail.com}
\end{array}
$
\end{center}

\maketitle \bigskip

%%%%%%%%%%%%%%%%%%%%%%%%%%%%%%%%%%%%%%%%%%%%%%%%%%%%%%%%%%%%%%%% Abstract %%%%%%%%%%%%%%%%%%%%%%%%%%%%%%%%%%%%%%%%%%%%%%%%%%%%%%%%%%%%

%\begin{abstract}

\noindent\textbf{Abstract}. In this paper we provide a \emph{global}
investigation of the geometry of parallelizable manifolds  (or
absolute parallelism geometry) frequently used for application. We
discuss the different linear connections and curvature tensors from
a global point of view. We give an existence and uniqueness theorem
for a remarkable linear connection, called the canonical connection.
Different curvature tensors are expressed in a compact form in terms
of the torsion tensor of the canonical connection only. Using the
Bianchi identities, some interesting identities are derived. An
important special fourth order tensor, which we refer to as Wanas
tensor, is globally defined and investigated. Finally a ``double-view"
for the fundamental geometric objects of an absolute parallelism
space is established: The expressions of these geometric objects are
computed in the parallelization basis and are compared with the
corresponding local expressions in the natural basis. Physical
aspects of some geometric objects considered are pointed out.
%\end{abstract}



\vspace{5pt} \noindent{\bf PACS 2010}: 02.04.Hw, 45.10.Na, 04.20.-q,
04.50.-h.

\newpage

%%%%%%%%%%%%%%%%%%%%%%%%%%%%%%%%%%%%%%%%%%%%%%%%%%%%%%%%%%%%%%% Section0: Introduction %%%%%%%%%%%%%%%%%%%%%%%%%%%%%%%%%%%%%%%%%%%%%%%%%%%%%%%%%%%%%%%

\section{Introduction}

\hspace*{.4cm} At the beginning of the $20^{th}$ century, the importance of
geometry in physical applications has been illuminated by Albert
Einstein. He has advocated a new philosophy known as ``The
Geometerization Philosophy". This philosophy can be summarized in
the following statement: ``To understand nature, one has to start
with geometry and end with physics" \cite{Wanasphilo}. In 1915,
Einstein used this philosophy to understand the essence of Gravity,
starting with a 4-dimensional Riemannian geometry, ending with a
successful theory for gravity; the General theory of Relativity (GR)
\cite{Einstein}. After the success of the theory, by testing its
predictions and applications, many authors have directed their
attention to the use of geometry to solve physical problems.

Einstein in his continuous attempts to understand more physical
interactions, has searched for a wider geometry to unify gravity and
electromagnetism. The problem with Riemannian geometry, however, is
that it has only ten degrees of freedom (the components of the
metric tensor in four dimensions) which are just sufficient to
describe gravity. Thus to construct a successful geometric theory
that would encompass both gravity and electromagnetism, one needs to
enlarge the number of degrees of freedom. This can be done in two
different ways: either by increasing the dimension of the underlying
space (\`{a} la-Kaluza-Klein) or by replacing the Riemannian
structure by another geometric structure having more degrees of
freedom (without increasing the dimension of the underlying space).
In tackling the problem of unification, Einstein has chosen the
second alternative. This led him to consider Absolute Parallelism
geometry (AP-geomety) \cite{AP-Ein} which has sixteen degrees of
freedom (the number of components of the vector fields forming the
parallelization); six extra degrees of freedom are gained. Many
developments of AP-geometry have been achieved  (e.g.,
\cite{unificationT, Wanashis, local}). Theories
constructed in this geometry (e.g., \cite{unificationP, Moller, unificationT}) together with applications
(e.g., \cite{Nashed1, Nashed2, Wanas}) show the
advantages of using AP-geometry in physics. Moreover, absolute parallelism characterizes the generalized Berwald spaces
among the Finsler spaces \cite{Tamassy1, Tamassy2}.

\vspace{6pt} In this paper, we establish a global approach to
AP-geometry. The global formulation of the different geometric
aspects of AP-geometry has many advantages. Some advantages of the
global formalism are:
\begin{itemize}
\item It could give more insight into the infra-structure of physical
theories constructed in the context of AP-geometry. Moreover, it may
offer the opportunity to unify field theories in a more economic
scheme.

\item It helps better understand the meaning and the essence of the
geometric objects and formulae without being trapped into the
complexity of indices. As a consequence, it reduces the probability
of mistake

\item It connects AP-geometry with the modern language of the
differential geometry.

\item In local coordinates some important expressions, such as the
Lie bracket $[\frac{\partial}{\partial x^i},\frac{\partial}{\partial
x^j}]$, disappear. Consequently, the contribution, geometrical or
physical, of all Lie brackets are completely hidden. Such
expressions do not vanish in global formalism. This may produce new
geometric or physical information.

\item The local formalism represents roughly a \textbf{micro}
viewpoint or a micro approach whereas the global formalism
represents a \textbf{macro} viewpoint. The two viewpoints are not
alternatives but rather complementary and are indispensable both for
geometry and physics.

\item As global results hold on the entire manifold (not only
on coordinate neighborhoods), they also hold locally. The converse
is not true; a result may hold locally but not globally. Moreover,
one can easily shift from global to local; it suffices to view the
global result in a coordinate neighborhood.
\end{itemize}

These are the main motivations of the present work, where all
results obtained are formulated in a prospective modern coordinate
free form.

\vspace{6pt} The paper is organized in the following manner. In
section 1, we define globally the basic elements of the AP-geometry
and prove an existence and uniqueness theorem for a remarkable
linear connection which we call the canonical connection \big(a
(flat) connection for which the parallelization vector fields are
parallel\big). We also study some properties of this connection. In
section 2, we define three other natural connections (the dual,
symmetric and Levi-Civita connections) and investigate their
properties together with the tensor fields associated to them. In
section 3, we express the curvature tensors of the above mentioned
three connections in a simple and compact form in terms of the
torsion tensor of the canonical connection only. We then use the
Bianchi identities to derive some further interesting identities. In
section 4, we give a global treatment of the W-tensor and
investigate some of its properties. In section 5, we present a
double-view for the fundamental geometric objects of AP-geometry: On
one hand, we consider the local expressions of these geometric
objects in the natural basis and, on the other hand, we compute
their expressions in the parallelization basis, and then compare
between the two sets
of expressions.\\
It should finally be noted that this work is based
mainly on \cite{local}.

Throughout the present paper we use the following notation:\\
$M$: an n-dimensional smooth real manifold,\\
$\mathfrak{F}(M)$: the $\mathds{R}$-algebra of $C^{\infty}$ functions on $M$,\\
$\mathfrak{X}(M)$: the $\mathfrak{F}(M)$-module of vector fields on $M$,\\
$T_xM$: the tangent space to $M$ at $x\in M$,\\
${T_x}^*M$: the cotangent space to $M$ at $x\in M$.\\
We make the assumption that all geometric objects we consider are of
class $C^{\infty}$.

%%%%%%%%%%%%%%%%%%%%%%%%%%%%%%%%%%%%%%%%%%%%%%%%%% Section1: Canonical connection %%%%%%%%%%%%%%%%%%%%%%%%%%%%%%%%%%%%%%%%%%%%%%%%%%%%%%

\section{Canonical connection}

\hspace*{.4cm} In this section, we give the definition of an AP-space and prove an existence and uniqueness theorem for a remarkable linear connection, which we call the canonical connection. Also, we prove some properties concerning this connection.

\begin{definition}\cite{Brickell}
A parallelizable manifold is a pair $(M,\;\undersym{X}{i})$, where $M$ is an n-dimensional smooth manifold and $\;\undersym{X}{i}\,(i = 1,\,...,\,n)$ are  n independent vector fields defined globally on $M$. The vector fields $\;\undersym{X}{1},\,...,\;\undersym{X}{n}$ are said to form a parallelization on $M$.
\end{definition}

Such a space is also known in the literature as an \emph{Absolute Parallelism space} or a \emph{Teleparallel space}. For simplicity, we will rather use the expressions ``\emph{AP-space} and \emph{AP-geometry}".

Since $\;\undersym{X}{i}$ are n independent vector fields on $M$, $\{\;\undersym{X}{i}(x): i=1,\,...,\,n\}$ is a basis of $T_xM$ for every $x\in M$. Any vector field $Y\in\mathfrak{X}(M)$ can be written globally as $Y=Y^{i}\,\undersym{X}{i}$, where $Y^{i}\in\mathfrak{F}(M)$. Here we use the notation $Y^{i}$ to denote the components of $Y$ with respect to $\,\undersym{X}{i}$. Einstein summation convention will be applied on Latin indices whatever their position is (even if the two repeated indices are downword).

\begin{definition}
The n differential $1$-forms $\;\undersym{\Omega}{i}:\mathfrak{X}(M)\longrightarrow\mathfrak{F}(M)$ defined by
\begin{equation} \label{1form}
\;\undersym{\Omega}{i}(\;\undersym{X}{j})=\delta_{ij}
\end{equation}
are called the parallelization forms.
\end{definition}

\hspace*{-.6cm}Clearly, if $Y=Y^{i}\;\undersym{X}{i}$, then
\begin{equation} \label{base}
\undersym{\Omega}{i}(Y)=Y^{i},\qquad\;\undersym{\Omega}{i}(Y)\;\undersym{X}{i}=Y.
\end{equation}

It follows directly from (\ref{1form}) that $\{\;\undersym{\Omega}{i}_{x}=\;\undersym{\Omega}{i}|_{T_xM}:i=1,\,...,\,n\}$ is the dual basis of the parallelization basis $\{\;\undersym{X}{i}(x):i=1,\,...,\,n\}$ for every $x\in M$. We call $\{\;\undersym{\Omega}{i}_{x}:i=1,\,...,\,n\}$ the dual parallelization basis of ${T_x}^*M$. The parallelization forms $\;\undersym{\Omega}{i}$ are independent in the $\mathfrak{F}(M)$-module $\mathfrak{X}^*(M)$.

\begin{lemma}\label{AP-condition}
Let $D$ be a linear connection on $M$. The $D$-covariant derivative of $\;\undersym{\Omega}{i}$ vanishes if and only if the $D$-covariant derivative of $\;\undersym{X}{i}$ vanishes.
\end{lemma}

\begin{proof}
For every $Y,\,Z\in\mathfrak{X}(M)$, we have, by (\ref{base}) and (\ref{1form}),
$$(D_Y\;\undersym{\Omega}{i})(Z)=(D_Y\;\undersym{\Omega}{i})\big(\;\undersym{\Omega}{j}(Z)\;\undersym{X}{j}\big)= -\;\undersym{\Omega}{j}(Z)\;\undersym{\Omega}{i}(D_Y\;\undersym{X}{j}).$$
Consequently, by (\ref{base}),
$$\big((D_Y\;\undersym{\Omega}{i})(Z)\big)\;\undersym{X}{i}=-\;\undersym{\Omega}{j}(Z)D_Y\;\undersym{X}{j},$$
from which the result follows.
\end{proof}

\begin{theorem}
On an AP-space $(M,\;\undersym{X}{i})$, there exists a unique linear connection $\nabla$ for which the parallelization vector fields $\;\undersym{X}{i}$ are parallel.
\end{theorem}

\begin{proof}
First we prove the uniqueness. Assume that $\nabla$ is a linear connection satisfying the condition $\nabla\;\undersym{X}{i}=0$. For all $Y,\,Z\in\mathfrak{X}(M)$, we have
$$\nabla_YZ=\nabla_Y\big(\;\undersym{\Omega}{i}(Z)\;\undersym{X}{i}\big)=\;\undersym{\Omega}{i}(Z)\nabla_Y\;\undersym{X}{i}
+\big(Y\cdot\;\undersym{\Omega}{i}(Z)\big)\;\undersym{X}{i} =\big(Y\cdot\;\undersym{\Omega}{i}(Z)\big) \;\undersym{X}{i}.$$
Hence, the connection $\nabla$ is uniquely determined by the relation
\begin{equation} \label{canonical}
\nabla_YZ=\big(Y\cdot\;\undersym{\Omega}{i}(Z)\big)\;\undersym{X}{i}.
\end{equation}

To prove the existence, let $\nabla:\mathfrak{X}(M)\times\mathfrak{X}(M)\longrightarrow\mathfrak{X}(M)$ be defined by (\ref{canonical}).
We show that $\nabla$ is a linear connection on $M$.
In fact, let $Y,\,Y_1,\,Y_2,\,Z,\,Z_1,\,Z_2\in\mathfrak{X}(M),\;f\in\mathfrak{F}(M)$. It is clear that
$\nabla_{Y_1+Y_2}Z=\nabla_{Y_1}Z+\nabla_{Y_2}Z$ and $\nabla_Y(Z_1+Z_2)=\nabla_YZ_1+\nabla_YZ_2$. Moreover,
\begin{eqnarray*}
\nabla_{fY}Z&=&\big((fY)\cdot\;\undersym{\Omega}{i}(Z)\big)\;\undersym{X}{i}=f\big(Y\cdot\;\undersym{\Omega}{i}(Z)\big)\;\undersym{X}{i}=f\nabla_YZ, \\ \nabla_Y(fZ)&=&\big(Y\cdot\;\undersym{\Omega}{i}(fZ)\big)\;\undersym{X}{i}=\Big(Y\cdot\big(f\;\undersym{\Omega}{i}(Z)\big)\Big)\;\undersym{X}{i} \\
         &=&f\big(Y\cdot\;\undersym{\Omega}{i}(Z)\big)\;\undersym{X}{i} + (Y\cdot f)\;\undersym{\Omega}{i}(Z)\;\undersym{X}{i} \\
         &=&f\nabla_YZ+(Y\cdot f)Z,\;\text{by (\ref{base}) and (\ref{canonical})}.
\end{eqnarray*}
It remains to show that $\nabla$ satisfies the condition $\nabla\;\undersym{X}{i}=0$:
$$\nabla_Y\;\undersym{X}{j}=\big(Y\cdot\;\undersym{\Omega}{i}(\;\undersym{X}{j})\big)\;\undersym{X}{i}=(Y\cdot\delta_{ij})\;\undersym{X}{i}=0.$$
This completes the proof.
\end{proof}

As a consequence of Lemma \ref{AP-condition}, we also have $\nabla\;\undersym{\Omega}{i}=0$. Hence
\begin{equation}\label{AP-cond}
\nabla\;\undersym{X}{i}=0,\quad\quad\nabla\;\undersym{\Omega}{i}=0.
\end{equation}
This property is known (locally) in the literature as the AP-condition.

\begin{definition}
Let $(M,\;\undersym{X}{i})$ be an AP-space. The unique linear connection $\nabla$ on $M$ defined by (\ref{canonical}) will be called the canonical connection of $(M,\;\undersym{X}{i})$.
\end{definition}

The canonical connection is of crucial importance because almost all geometric objects in the AP-space will be built up of it, as will be seen throughout the paper.

\vspace*{.26cm}Now we give an intrinsic formula of the torsion tensor $T$ of $\nabla$.

\begin{proposition}
The torsion tensor $T$ of the canonical connection is given by
\begin{equation} \label{torsion}
T(Y,\,Z)=\;\undersym{\Omega}{i}(Y)\;\undersym{\Omega}{j}(Z)[\;\undersym{X}{j},\;\undersym{X}{i}].
\end{equation}
\end{proposition}

\begin{proof}
The torsion tensor $T$ of $\nabla$ is defined, for all $Y,\,Z\in\mathfrak{X}(M)$, by
$$T(Y,\,Z)=\nabla_YZ-\nabla_ZY-[Y,\,Z].$$
Using the AP-condition (\ref{AP-cond}), we get
\begin{eqnarray*}
T(Y,\,Z)&\overset{\;\text{(\ref{base})}}{=}&T\big(\;\undersym{\Omega}{i}(Y)\;\undersym{X}{i},\;\undersym{\Omega}{j}(Z)\;\undersym{X}{j}\big) =\;\undersym{\Omega}{i}(Y)\;\undersym{\Omega}{j}(Z)T(\;\undersym{X}{i},\;\undersym{X}{j})\\
        &=&\;\undersym{\Omega}{i}(Y)\;\undersym{\Omega}{j}(Z)(\nabla_{\;\undersym{X}{i}}\;\undersym{X}{j}-\nabla_{\;\undersym{X}{j}}\;\undersym{X}{i} -[\;\undersym{X}{i},\;\undersym{X}{j}])= \;\undersym{\Omega}{i}(Y)\;\undersym{\Omega}{j}(Z)[\;\undersym{X}{j},\;\undersym{X}{i}].
\end{eqnarray*}
\vspace*{-1.6cm}\[\qedhere\]
\end{proof}

\begin{theorem}\label{flat}
Let $(M,\;\undersym{X}{i})$ be an AP-space. The canonical connection of $(M,\;\undersym{X}{i})$ is flat.
\end{theorem}

\begin{proof}
The result follows from the definition of the curvature tensor $R$ of $\nabla$:
$$R(Y,\,Z)V=\nabla_Y\nabla_ZV-\nabla_Z\nabla_YV-\nabla_{[Y,Z]}V$$
and the AP-condition (\ref{AP-cond}).
\end{proof}

\begin{remark}\cite{local}
It is for this reason that many authors think that the AP-space is a flat space. This is by no means true. In fact, it is meaningless to speak
of curvature without reference to a connection. All we can say is that the AP-space is flat with respect to its canonical connection. However, there are other three natural connections on an AP-space which are nonflat, as will be shown later.
\end{remark}

%%%%%%%%%%%%%%%%%%%%%%%%%%%%%%%%%%%%%%%%%%%%%%%%%%%%%%%%%%%% section 2: Canonical Connection %%%%%%%%%%%%%%%%%%%%%%%%%%%%%%%%%%%%%%%%%%%%%%%%%%%%%%%%%%%%%

\section{Other linear connections on an AP-space}

\hspace*{.4cm} In this section, we define a metric on an AP-space and investigate the properties of three other natural connections on the space. Moreover, we define the contortion tensor and give its relation to the torsion tensor (\ref{torsion}).

\begin{theorem}
Let $(M,\;\undersym{X}{i})$ be an AP-space and $\;\undersym{\Omega}{i}$ the parallelization forms on $M$. Then
\begin{equation} \label{metric}
g:=\;\undersym{\Omega}{i}\otimes\;\undersym{\Omega}{i}
\end{equation}
defines a metric tensor on $M$.
\end{theorem}

\begin{proof}
Clearly $g$ is a symmetric tensor of type $(0,2)$ on $M$. For all $Y\in\mathfrak{X}(M)$, we have
$$g(Y,\,Y)=(\;\undersym{\Omega}{i}\otimes\;\undersym{\Omega}{i})(Y,Y)=\sum_{i=1}^{n}\big(\;\undersym{\Omega}{i}(Y)\big)^2\geq 0.$$
Moreover,
$$g(Y,Y)=0\Longrightarrow\sum_{i=1}^{n}\big(\;\undersym{\Omega}{i}(Y)\big)^2=0\Longrightarrow\;\undersym{\Omega}{i}(Y)=0\;\,\forall i\Longrightarrow\;\undersym{\Omega}{i}(Y)\;\undersym{X}{i}=0\overset{(\ref{base})}{\Longrightarrow} Y=0.$$
Hence, $g$ is a metric tensor on $M$.
\end{proof}

\begin{remark}\label{rem0}
It is clear that:
\begin{description}
  \item[(a)] $g(\;\undersym{X}{i},\;\undersym{X}{j})=\delta_{ij}.\hfill\refstepcounter{equation}(\theequation)\label{orthogonal}$
  \item[(b)] $g(\;\undersym{X}{i},Y)=\;\undersym{\Omega}{i}(Y).\hfill\refstepcounter{equation}(\theequation)\label{base2}$
\end{description}
\end{remark}

Property {\bf{(a)}} shows that the parallelization vector fields $\;\undersym{X}{i}$ are $g$-orthonormal and {\bf{(b)}} provides the duality between $\;\undersym{X}{i}$ and $\;\undersym{\Omega}{i}$ via $g$.

\begin{lemma}
Let $(M,\;\undersym{X}{i})$ be an AP-space. A linear connection $D$ on $M$ is a metric connection if and only if
\begin{equation*} %\label{metric1}
\;\undersym{\Omega}{i}(D_V\;\undersym{X}{j})+\;\undersym{\Omega}{j}(D_V\;\undersym{X}{i})=0.
\end{equation*}
\end{lemma}

\begin{proof}
By simple calculation, using (\ref{orthogonal}) and (\ref{base2}), one can show that
$$(D_Vg)(\;\undersym{X}{i},\;\undersym{X}{j})=-\;\undersym{\Omega}{i}(D_V\;\undersym{X}{j})-\;\undersym{\Omega}{j}(D_V\;\undersym{X}{i}),$$
from which the result follows.
\end{proof}

The last lemma together with the AP-condition (\ref{AP-cond}) give rise to the next result.

\begin{proposition}
The canonical connection is a metric connection.
\end{proposition}

\begin{proposition}
On an AP-space there are three other (built-in) linear connections:
\begin{description}
\item[(a)] The dual connection $\widetilde{\nabla}$ given by
  \begin{equation} \label{dualcon}
  \widetilde{\nabla}_Y Z:=\nabla_Z Y + [Y,Z].
  \end{equation}
\item[(b)] The symmetric connection $\widehat{\nabla}$ given by
  \begin{equation} \label{symcon}
  \widehat{\nabla}_YZ:=\frac{1}{2}(\nabla_YZ+\nabla_ZY+[Y,Z]).
  \end{equation}
\item[(c)] The Levi-Civita connection $\oversetc{\nabla}$ is given by \cite{Kuhnel}
  \begin{eqnarray} \label{riemannian}
  2g(\oversetc{\nabla}_YZ,\,V)
  &=&Y\cdot\,g(Z,\,V)+Z\cdot\,g(V,\,Y)-V\cdot\,g(Y,\,Z) \nonumber \\
  & &-g(Y,\,[Z,\,V])+g(Z,\,[V,\,Y])+g(V,\,[Y,\,Z]).
  \end{eqnarray}
\end{description}
\end{proposition}

The proof is straightforward and we omit it.

\begin{remark}\label{rem1}
One can easily show that:
\begin{description}
  \item[(a)] $\widetilde{\nabla}_YZ=\nabla_YZ-T(Y,\,Z)$.
  \item[(b)] $\widehat{\nabla}_YZ=\nabla_YZ-\frac{1}{2}T(Y,\,Z)=\frac{1}{2}(\nabla_YZ+\widetilde{\nabla}_YZ)$.
  \item[(c)] $\widehat{\nabla}$ and $\oversetc{\nabla}$ are torsionless whereas $\nabla$ and $\widetilde{\nabla}$ have the same torsion up to a sign.
\end{description}
Here $T$ is the torsion tensor of the canonical connection $\nabla$. Since there are no other torsion tensors in the space, we can say that $T$ is the torsion of the space.
\end{remark}

In Reimannian geometry the Levi-Civita connection has no explicit expression. However, in AP-geometry we can have an \emph{explicit expression} for the Levi-Civita connection $\oversetc{\nabla}$ as shown in the following.

\begin{theorem}
Let $(M,\,\;\undersym{X}{i})$ be an AP-space. Then the Levi-Civita connection $\oversetc{\nabla}$ can be written in the form:
\begin{equation} \label{riemannian2}
\oversetc{\nabla}_YZ=\widehat{\nabla}_YZ-\frac{1}{2}(\mathcal{L}_{\;\undersym{X}{i}}g)(Y,\,Z)\;\undersym{X}{i},
\end{equation}
where $\mathcal{L}_Y$ is the Lie derivative with respect to $Y\in\mathfrak{X}(M)$.
\end{theorem}

\vspace{0pt}

\begin{proof}
By replacing $V$ in (\ref{riemannian}) by $\;\undersym{X}{i}$ and using (\ref{base2}), we get
$$2\;\undersym{\Omega}{i}(\oversetc{\nabla}_YZ)=Y\cdot\;\undersym{\Omega}{i}(Z)+Z\cdot\;\undersym{\Omega}{i}(Y)-\;\undersym{X}{i}\cdot g(Y,\,Z)+g(Y,\,[\;\undersym{X}{i},\,Z])+g(Z,\,[\;\undersym{X}{i},\,Y])+\;\undersym{\Omega}{i}([Y,\,Z]).$$
Taking into account (\ref{base}) and (\ref{canonical}), the above equation reads
\begin{eqnarray*}
2\oversetc{\nabla}_YZ&=&\nabla_YZ+\nabla_ZY-\big(\;\undersym{X}{i}\cdot g(Y,\,Z)\big)\;\undersym{X}{i}+g(Y,\,[\;\undersym{X}{i},\,Z])\;\undersym{X}{i}+g(Z,\,[\;\undersym{X}{i},\,Y])\;\undersym{X}{i}+[Y,\,Z] \\
&=&2\widehat{\nabla}_YZ-\Big(\;\undersym{X}{i}\cdot g(Y,\,Z)-g(Y,\,[\;\undersym{X}{i},\,Z])-g(Z,\,[\;\undersym{X}{i},\,Y])\Big)\;\undersym{X}{i},\;\text{by (\ref{symcon})} \\
&=&2\widehat{\nabla}_YZ-(\mathcal{L}_{\;\undersym{X}{i}}g)(Y,\,Z)\;\undersym{X}{i}.
\end{eqnarray*}
\vspace*{-1.6cm}\[\qedhere\]
\end{proof}

\begin{corollary}
In an AP-space, the Levi-Civita connection and the symmetric connection coincide if, and only if, the parallelization vector fields are Killing vector fields:
$$\oversetc{\nabla}=\widehat{\nabla}\Longleftrightarrow\mathcal{L}_{\;\undersym{X}{i}}g=0\;\,\forall i.$$
\end{corollary}

\begin{definition}
The contortion tensor $C$ is defined by the formula:
\begin{equation} \label{contortion1}
C(Y,\,Z) = \nabla_Y Z - \oversetc{\nabla}_Y Z.
\end{equation}
\end{definition}

The contortion tensor may also be written in the form:
\begin{equation} \label{contortion2}
C(Y,\,Z)=(\oversetc{\nabla}_Y\;\undersym{\Omega}{i})(Z)\;\undersym{X}{i}.
\end{equation}
In fact, using (\ref{base}) and (\ref{canonical}), we have for all $Y,\,Z\in\mathfrak{X}(M)$,
$$C(Y,\,Z)=\nabla_YZ-\oversetc{\nabla}_YZ=\big(Y\,\cdot\;\undersym{\Omega}{i}\;(Z)\big)\;\undersym{X}{i}-\;\undersym{\Omega}{i}(\oversetc{\nabla}_YZ)\;\undersym{X}{i}=(\oversetc{\nabla}_Y\;\undersym{\Omega}{i})(Z)\;\undersym{X}{i}.$$

The identities (\ref{contortion1}) and (\ref{contortion2}) show that the geometry of an AP-space can be built up from the Levi-Civita connection instead of the canonical connection:
$$\nabla_YZ=\oversetc{\nabla}_YZ+(\oversetc{\nabla}_Y\;\undersym{\Omega}{i})(Z)\;\undersym{X}{i}.$$

The next proposition establishes the mutual relations between the torsion and contortion tensors.

\begin{proposition}
The following identities hold:
\begin{description}\label{rtc}
  \item[(a)]$T(Y,\,Z)=C(Y,\,Z)-C(Z,\,Y). $%\hfill\refstepcounter{equation}(\theequation)\label{TC}$
  \item[(b)]$C(Y,\,Z)=\frac{1}{2}\Big(T(Y,\,Z)+T(\;\undersym{X}{i},\,Y,\,Z)\;\undersym{X}{i}+T(\;\undersym{X}{i},\,Z,\,Y)\;\undersym{X}{i}\Big).$ %\hfill\refstepcounter{equation}(\theequation)\label{contortion3}$
\end{description}
From which,
\begin{description}
  \item[(a)$'$]$T(Y,\,Z,\,V)=C(Y,\,Z,\,V)-C(Z,\,Y,\,V).$ % \hfill\refstepcounter{equation}(\theequation)\label{T3}$
  \item[(b)$'$]$C(Y,\,Z,\,V)=\frac{1}{2}\Big(T(Y,\,Z,\,V)+T(V,\,Y,\,Z)+T(V,\,Z,\,Y)\Big),$ %\hfill\refstepcounter{equation}(\theequation)\label{CT}$
\end{description}
where $C(Y,\,Z,\,V)=g\big(C(Y,\,Z),\,V\big)$ and $T(Y,\,Z,\,V)=g\big(T(Y,\,Z),\,V\big)$.\\ Consequently, the torsion tensor vanishes if and only if the contortion tensor vanishes.
\end{proposition}

\begin{proof}
Let $Y,\,Z,\,V\in\mathfrak{X}(M)$. Then,
\begin{description}
  \item[(a)] The first identity gives the torsion tensor in terms of the contortion tensor.
  \begin{eqnarray*}
  T(Y,\,Z)&=&\nabla_YZ-\nabla_ZY-[Y,\,Z] \\
          &=&(\nabla_Y Z -\nabla_Z Y)-(\oversetc{\nabla}_Y Z-\oversetc{\nabla}_Z Y),\;\text{since $\oversetc{\nabla}$ is torsionless} \\
          &=& (\nabla_Y Z - \oversetc{\nabla}_Y Z)-(\nabla_Z Y - \oversetc{\nabla}_Z Y)=C(Y,\,Z)-C(Z,\,Y).
  \end{eqnarray*}
  \item[(b)] The second identity gives the contortion tensor in terms of the torsion tensor. In the following proof we make use of (\ref{riemannian2}), Remark \ref{rem1}, (\ref{base}), Remark \ref{rem0} and (\ref{torsion}).
  \begin{eqnarray*}
  2C(Y,\,Z)&=&2\nabla_YZ-2\oversetc{\nabla}_YZ=2\nabla_YZ-2\widehat{\nabla}_YZ+(\mathcal{L}_{\;\undersym{X}{i}}g)(Y,\,Z)\;\undersym{X}{i}\\
           &=&2\nabla_YZ-2\nabla_YZ+T(Y,\,Z)+\;\undersym{\Omega}{j}(Y)\;\undersym{\Omega}{k}(Z)(\mathcal{L}_{\;\undersym{X}{i}}g)(\;\undersym{X}{j}, \,\;\undersym{X}{k})\;\undersym{X}{i}\\
           &=&T(Y,\,Z)+\;\undersym{\Omega}{j}(Y)\;\undersym{\Omega}{k}(Z)\Big(\;\undersym{X}{i}\cdot g(\;\undersym{X}{j},\;\undersym{X}{k})-g([\;\undersym{X}{i},\;\undersym{X}{j}],\;\undersym{X}{k}) -g([\;\undersym{X}{i},\;\undersym{X}{k}],\;\undersym{X}{j})\Big)\;\undersym{X}{i} \\
           %&=&T(Y,\,Z)-\;\undersym{\Omega}{j}(Y)\;\undersym{\Omega}{k}(Z)\Big(g\big(T(\;\undersym{X}{j},\;\undersym{X}{i}),\;\undersym{X}{k}\big)+ g\big(T(\;\undersym{X}{k},\;\undersym{X}{i}),\;\undersym{X}{j}\big)\Big)\;\undersym{X}{i} \\
           &=&T(Y,\,Z)-\Big(g\big(T(Y,\;\undersym{X}{i}),\,Z\big)+g\big(T(Z,\;\undersym{X}{i}),\,Y\big)\Big)\;\undersym{X}{i} \\
           &=&T(Y,\,Z)+\Big(T(\;\undersym{X}{i},\,Y,\,Z)+T(\;\undersym{X}{i},\,Z,\,Y)\Big)\;\undersym{X}{i}.
  \end{eqnarray*}
\end{description}
\vspace*{-1.2cm}\[\qedhere\]
\end{proof}
\vspace*{.26cm}
\begin{remark}
$T(Y,\,Z,\,V)$ is skew-symmetric in the first two arguments whereas $C(Y,\,Z,\,V)$ is skew-symmetric in the last two arguments.
\end{remark}

\begin{definition}
Let $(M,\;\undersym{X}{i})$ be an AP-space. The contracted torsion or the basic form $B$ is defined, for every $Y\in\mathfrak{X}(M)$ by
\begin{equation*}
B(Y):={\rm Tr}\{Z\longmapsto T(Z,\,Y)\}.
\end{equation*}
\end{definition}

This $1$-form is known (locally) in the literature as the basic vector. In terms of the metric tensor (\ref{metric}), using (\ref{orthogonal}), the basic form can be written as
\begin{equation} \label{basicvector}
B(Y)=g\big(T(\;\undersym{X}{i},\,Y),\;\undersym{X}{i}\big)=T(\;\undersym{X}{i},\,Y,\;\undersym{X}{i}).
\end{equation}
Using Proposition \ref{rtc}{\bf{(b)$'$}}, $B(Y)$ can also be expressed in the form
\begin{equation*} %\label{basicvector3}
B(Y)=C(\;\undersym{X}{i},\,Y,\;\undersym{X}{i}).
\end{equation*}
Making use of (\ref{basicvector}) and (\ref{torsion}), we have
\begin{equation*} %\label{basicvector2}
B(Y)=\;\undersym{\Omega}{j}(Y)\;\undersym{\Omega}{i}([\;\undersym{X}{j},\;\undersym{X}{i}]).
\end{equation*}

\emph{It should be noted that in the above three expressions and in similar expressions summation is carried out on repeated mesh indices, although they are situated in different argument positions.
}

\begin{proposition} \label{different-connections}
Concerning the four connections of the AP-space, the difference tensors are given by:
\begin{description}
  \item[(a)] $\nabla_YZ-\widetilde{\nabla}_YZ=T(Y,\,Z)$.
  \item[(b)] $\nabla_YZ-\widehat{\nabla}_YZ=\frac{1}{2}T(Y,\,Z)$.
  \item[(c)] $\nabla_YZ-\oversetc{\nabla}_YZ=C(Y,\,Z)$.
  \item[(d)] $\widetilde{\nabla}_YZ-\widehat{\nabla}_YZ=-\frac{1}{2}T(Y,\,Z)$.
  \item[(e)] $\widetilde{\nabla}_YZ-\oversetc{\nabla}_YZ=C(Z,\,Y)$.
  \item[(f)] $\widehat{\nabla}_YZ-\oversetc{\nabla}_YZ=\frac{1}{2}(\mathcal{L}_{\;\undersym{X}{i}}g)(Y,\,Z)\;\undersym{X}{i}.$
\end{description}
\end{proposition}

\begin{proof}
Properties {\bf{(a)}}, {\bf{(b)}}, {\bf{(d)}} follow from Remark \ref{rem1}, {\bf{(c)}} is the definition of the contortion tensor, {\bf{(e)}} follows from (\ref{contortion1}) and the fact that $\oversetc{\nabla}$ is torsionless, and {\bf{(f)}} follows from (\ref{riemannian2}).
\end{proof}

As a consequence of the above proposition, we have the following useful relations.

\begin{corollary} \label{difofconnections}
For every $Y,\,Z,\,V\in\mathfrak{X}(M)$, we have the following relations:
\begin{description}
  \item[(a)] $(\nabla_VT)(Y,\,Z)-(\widetilde{\nabla}_VT)(Y,\,Z)=\underset{Y,\,Z,\,V}{\mathfrak{S}}\;\Big\{T\big(V,\,T(Y,\,Z)\big)\Big\}.$ %\hfill\refstepcounter{equation}(\theequation)\label{canonical-dual}$
  \item[(b)] $(\nabla_VT)(Y,\,Z)-(\widehat{\nabla}_VT)(Y,\,Z)=\frac{1}{2}\;\underset{Y,\,Z,\,V}{\mathfrak{S}}\;\Big\{T\big(V,\,T(Y,\,Z)\big)\Big\}.$ %\hfill\refstepcounter{equation}(\theequation)\label{canonical-sym}$
  \item[(c)] $(\nabla_VT)(Y,\,Z)-(\oversetc{\nabla}_VT)(Y,\,Z)=-T\big(Y,\,C(V,\,Z)\big)+T\big(Z,\,C(V,\,Y)\big)+C\big(V,\,T(Y,\,Z)\big),$ %\hfill\refstepcounter{equation}(\theequation)\label{canonical-riem}$
\end{description}
where $\underset{Y,\,Z,\,V}{\mathfrak{S}}$ denotes the cyclic permutation of $Y,\,Z,\,V$ and summation.
\end{corollary}

%%%%%%%%%%%%%%%%%%%%%%%%%%%%%%%%%%%%%%%%%%%%% Section3: Curvature Tensors and Bianchi Identities %%%%%%%%%%%%%%%%%%%%%%%%%%%%%%%%%%%%%%%%%%%%%%%%%%

\section{Curvature tensors and Bianchi identities}

\vspace{4pt}

\hspace*{.4cm} In an AP-space the curvature $R$ of the canonical connection $\nabla$ vanishes identically. This section is devoted to show that the other three curvature tensors $\widetilde{R}, \widehat{R}$ and $\;\overcirc{R}$, associated with $\widetilde{\nabla}, \widehat{\nabla}$ and $\oversetc{\nabla}$ respectively, do not vanish. Also, we show that the vanishing of $R$ enables us to express these three curvature tensors in terms of the torsion tensor only.

\begin{theorem}
The three curvature tensors $\widetilde{R}, \widehat{R}$ and $\;\overcirc{R}$ of the connections $\widetilde{\nabla}, \widehat{\nabla}$ and $\oversetc{\nabla}$ are given respectively by:
\begin{description}
  \item[(a)] $\;\widetilde{R}(Y,\,Z)V=(\nabla_VT)(Y,\,Z).\hfill\refstepcounter{equation}(\theequation)\label{dual}$
\end{description}
\vspace*{8pt}
{\bf{(b)}}
\vspace*{-0.9cm}
\begin{eqnarray} \label{sym}                                         \hspace*{-25.5cm}\widehat{R}(Y,\,Z)V&=&\frac{1}{2}\Big((\nabla_ZT)(Y,\,V)-(\nabla_YT)(Z,\,V)\Big)-\frac{1}{2}T\big(T(Y,\,Z),\,V\big)\quad\quad\quad\quad\nonumber\\
              & &+\frac{1}{4}\Big(T\big(Y,\,T(Z,\,V)\big)-T\big(Z,\,T(Y,\,V)\big)\Big).
\end{eqnarray}
{\bf{(c)}}
\vspace*{-1.0cm}
\begin{eqnarray} \label{Riemanniancurv}
 \hspace*{-35cm}\overcirc{R}(Y,\,Z)V&=&(\nabla_Z C)(Y,\,V)-(\nabla_Y C)(Z,\,V)-C\big(T(Y,\,Z),\,V\big)\quad\quad\quad\quad\quad\quad\;\; \nonumber \\
                   & &+C\big(Y,\,C(Z,\,V)\big)-C\big(Z,\,C(Y,\,V)\big).
\end{eqnarray}
\end{theorem}

\begin{proof}
We prove {\bf{(a)}} only. The proof of the other parts can be carried out in the same manner. Using (\ref{dualcon}), we get
\begin{eqnarray*}
\widetilde{\nabla}_Y\;\widetilde{\nabla}_ZV&=&\widetilde{\nabla}_Y(\nabla_VZ+[Z,\,V])=\widetilde{\nabla}_Y\nabla_VZ+\widetilde{\nabla}_Y[Z,\,V] \\
                       &=&\nabla_{\nabla_VZ}Y+[Y,\,\nabla_VZ]+\nabla_{[Z,\,V]}Y+[Y,\,[Z,\,V]].
\end{eqnarray*}
Similarly,
$$\widetilde{\nabla}_Z\;\widetilde{\nabla}_YV=\nabla_{\nabla_VY}Z+[Z,\,\nabla_VY]+\nabla_{[Y,\,V]}Z+[Z,\,[Y,\,V]].$$
and
$$\widetilde{\nabla}_{[Y,\,Z]}V=\nabla_V[Y,\,Z]+[[Y,\,Z],\,V].$$
Using the above three identities, together with the Jacobi identity, we get
\begin{eqnarray*}
\widetilde{R}(Y,\,Z)V&=&\widetilde{\nabla}_Y\widetilde{\nabla}_ZV-\widetilde{\nabla}_Z\widetilde{\nabla}_YV-\widetilde{\nabla}_{[Y,\,Z]}V \\
              &=&\nabla_{\nabla_VZ}Y+[Y,\,\nabla_VZ]-\nabla_{\nabla_VY}Z-[Z,\,\nabla_VY] \\
              & &-\nabla_V[Y,\,Z]+\nabla_{[Z,\,V]}Y-\nabla_{[Y,\,V]}Z.
\end{eqnarray*}
Using the fact that the curvature tensor of the canonical connection vanishes (Theorem \ref{flat}), it follows that
$$\nabla_{[Y,\,Z]}V=\nabla_Y\nabla_ZV-\nabla_Z\nabla_YV.$$
Using the above identity, we get
\begin{eqnarray*}
\widetilde{R}(Y,\,Z)V&=&\nabla_{\nabla_VZ}Y+[Y,\,\nabla_VZ]-\nabla_{\nabla_VY}Z-[Z,\,\nabla_VY]-\nabla_V[Y,\,Z] \\
              & &+\nabla_Z\nabla_VY-\nabla_V\nabla_ZY-\nabla_Y\nabla_VZ+\nabla_V\nabla_YZ\\
              &=&(\nabla_V\nabla_YZ-\nabla_V\nabla_ZY-\nabla_V[Y,\,Z]) \\
              & &-(\nabla_{\nabla_VY}Z-\nabla_Z\nabla_VY-[\nabla_VY,\,Z]) \\
              & &-(\nabla_Y\nabla_VZ-\nabla_{\nabla_VZ}Y-[Y,\,\nabla_VZ]) \\
              &=&\nabla_VT(Y,\,Z)-T(\nabla_VY,\,Z)-T(Y,\,\nabla_VZ) \\
              &=&(\nabla_VT)(Y,\,Z).
\end{eqnarray*}
\vspace*{-1.6cm}\[\qedhere\]
\end{proof}

The above theorem shows that the curvature tensors $\widetilde{R},\,\widehat{R}$ and $\;\overcirc{R}$ are expressible in terms of the torsion tensor of the space only. This proves that the geometry of an AP-space depends crucially on the torsion tensor. It is worth mentioning that the vanishing of that tensor implies that the four connections $\nabla,\,\widetilde{\nabla},\,\widehat{\nabla}$ and $\oversetc{\nabla}$ coincide and a trivial flat Riemannian space is achieved. Thus, a sufficient condition for the non-vanishing of the torsion tensor is the non-vanishing of any one of the three curvature tensors $\widetilde{R},\,\widehat{R}$ or $\;\overcirc{R}$.

\vspace*{.26cm}The next result gives the expressions of the Ricci tensor ${\rm\oversetc{R}ic}$ of $\oversetc{\nabla}$ and the Ricci-like tensors ${\rm\widetilde{R}ic}$ and ${\rm\widehat{R}ic}$ of $\widetilde{\nabla}$ and $\widehat{\nabla}$ together with their respective contractions (the scalar curvature ${\rm\oversetc{S}c}$ and the curvature-like scalars ${\rm\widetilde{S}c}$ and ${\rm\widehat{S}c}$). The orthonormality of the parallelization vector fields $\;\undersym{X}{i}$ plays an essential role in the proof.

\begin{theorem}\label{exception}
In an AP-space $(M,\;\undersym{X}{i})$ we have, for every $Y,\,Z\in\mathfrak{X}(M)$,
\begin{description}
  \item[(a)] ${\rm\widetilde{R}ic}(Y,\,Z)=-(\nabla_ZB)(Y).$ %\hfill\refstepcounter{equation}(\theequation)\label{dual-ricci}$
  \item[(b)] ${\rm\widehat{R}ic}(Y,\,Z)=\frac{1}{2}(\mathcal{L}_{\;\undersym{X}{i}}T)(Y,\,Z,\;\undersym{X}{i}) +\frac{1}{4}T\big(Y,\,T(Z,\;\undersym{X}{i}),\;\undersym{X}{i}\big)-\frac{1}{2}(\nabla_YB)(Z)- \frac{1}{4}B\big(T(Y,\,Z)\big).$%\hfill\refstepcounter{equation}(\theequation)\label{sym-ricci}$
  \item[(c)] ${\rm\oversetc{R}ic}(Y,\,Z)=(\mathcal{L}_{\;\undersym{X}{i}}C)(Y,\,Z,\;\undersym{X}{i}) +C(Y,\,C\big(Z,\;\undersym{X}{i}),\;\undersym{X}{i}\big)-(\nabla_YB)(Z)-B(C(Y,\,Z)).$ %\hfill\refstepcounter{equation}(\theequation)\label{riem-ricci}$
  \item[(a)$'$] ${\rm\widetilde{S}c}=-\;\undersym{X}{i}\cdot B(\;\undersym{X}{i}).$
  \item[(b)$'$] ${\rm\widehat{S}c}=-\frac{1}{2}\;\undersym{X}{i}\cdot B(\;\undersym{X}{i})+\frac{1}{4}T(\;\undersym{X}{j},\,[\;\undersym{X}{i},\;\undersym{X}{j}],\;\undersym{X}{i}).$
  \item[(c)$'$] ${\rm\oversetc{S}c}=-2\;\undersym{X}{i}\cdot B(\;\undersym{X}{i})+B(\;\undersym{X}{i})B(\;\undersym{X}{i})+C(T(\;\undersym{X}{i},\;\undersym{X}{j}),\;\undersym{X}{j},\;\undersym{X}{i}) +C\big(\;\undersym{X}{j},\,C(\;\undersym{X}{i},\;\undersym{X}{j}),\;\undersym{X}{i}\big)$.
\end{description}
\end{theorem}
\begin{proof} We prove {\bf{(b)}} and {\bf{(c)$'$}} only. The other
identities can be proved similarly.
\begin{description}
  \item [(b)] Using (\ref{orthogonal}), (\ref{sym}), (\ref{AP-cond}) and (\ref{basicvector}), we have
  \begin{eqnarray*}
  {\rm\widehat{R}ic}(Y,\,Z)&=&g\big(\widehat{R}(Y,\;\undersym{X}{i})Z,\;\undersym{X}{i}\big)\\
                 &=&\frac{1}{2}\Big((\nabla_{\;\undersym{X}{i}}T)(Y,\,Z,\;\undersym{X}{i})-(\nabla_YT)(\;\undersym{X}{i},\,Z,\;\undersym{X}{i})- T\big(T(Y,\;\undersym{X}{i}),\,Z,\;\undersym{X}{i}\big)\Big)\\
                 & &+\frac{1}{4}\Big(T\big(Y,\,T(\;\undersym{X}{i},\,Z),\;\undersym{X}{i}\big) -T\big(\;\undersym{X}{i},\,T(Y,\,Z),\;\undersym{X}{i}\big)\Big)\\
                 &=&\frac{1}{4}\Big(2\;\undersym{X}{i}\cdot
                     T(Y,\,Z,\;\undersym{X}{i})-2T(\nabla_{\;\undersym{X}{i}}Y,\,Z,\;\undersym{X}{i}) -2T(Y,\,\nabla_{\;\undersym{X}{i}}Z,\;\undersym{X}{i})-2(\nabla_YB)(Z)\\
                 & &-2T\big(T(Y,\;\undersym{X}{i}),\,Z,\;\undersym{X}{i}\big)+T\big(Y,\,T(\;\undersym{X}{i},\,Z),\;\undersym{X}{i}\big) -B\big(T(Y,\,Z)\big)\Big)\\
                 &=&\frac{1}{4}\Big(2\;\undersym{X}{i}\cdot T(Y,\,Z,\;\undersym{X}{i})-T(Y,\,\nabla_{\;\undersym{X}{i}}Z,\;\undersym{X}{i})-2(\nabla_YB)(Z)\\
                 & &-2T([\;\undersym{X}{i},\,Y],\,Z,\;\undersym{X}{i})-T(Y,\,[\;\undersym{X}{i},\,Z],\;\undersym{X}{i})-B\big(T(Y,\,Z)\big)\Big)\\
                 &=&\frac{1}{4}\Big(2(\mathcal{L}_{\;\undersym{X}{i}}T)(Y,\,Z,\;\undersym{X}{i})-T(Y,\,\nabla_{\;\undersym{X}{i}}Z,\;\undersym{X}{i}) -T(Y,\,[Z,\;\undersym{X}{i}],\;\undersym{X}{i})\\
                 & &-2(\nabla_YB)(Z)-B\big(T(Y,\,Z)\big)\Big)\\
                 &=&\frac{1}{2}(\mathcal{L}_{\;\undersym{X}{i}}T)(Y,\,Z,\;\undersym{X}{i}) +\frac{1}{4}T\big(Y,\,T(Z,\;\undersym{X}{i}),\;\undersym{X}{i}\big)-\frac{1}{2}(\nabla_YB)(Z)- \frac{1}{4}B\big(T(Y,\,Z)\big).
  \end{eqnarray*}
  \item [(c)$'$] Using (\ref{orthogonal}), {\bf{(c)}}, (\ref{AP-cond}), (\ref{contortion1}), Proposition \ref{rtc}{\bf{(b)}} and the torsionless property of $\oversetc{\nabla}$, we get
  \begin{eqnarray*}
  {\rm\oversetc{S}c}&=&{\rm\oversetc{R}ic}(\;\undersym{X}{j},\;\undersym{X}{j})\\
         &=&(\mathcal{L}_{\;\undersym{X}{i}}C)(\;\undersym{X}{j},\;\undersym{X}{j},\;\undersym{X}{i}) +C\big(\;\undersym{X}{j},\,C(\;\undersym{X}{j},\;\undersym{X}{i}),\;\undersym{X}{i}\big)- (\nabla_{\;\undersym{X}{j}}B)(\;\undersym{X}{j})-B\big(C(\;\undersym{X}{j},\;\undersym{X}{j})\big)\\
         &=&\;\undersym{X}{i}\cdot C(\;\undersym{X}{j},\;\undersym{X}{j},\;\undersym{X}{i})-C([\;\undersym{X}{i},\;\undersym{X}{j}],\;\undersym{X}{j},\;\undersym{X}{i}) -C(\;\undersym{X}{j},\,[\;\undersym{X}{i},\;\undersym{X}{j}],\;\undersym{X}{i})- \;\undersym{X}{j}\cdot B(\;\undersym{X}{j})\\
         & &-C(\;\undersym{X}{j},\,\oversetc{\nabla}_{\;\undersym{X}{j}}\;\undersym{X}{i},\;\undersym{X}{i})+ B(\;\undersym{X}{i})B(\;\undersym{X}{i})\\
&=&-2\;\undersym{X}{i}\cdot B(\;\undersym{X}{i})+B(\;\undersym{X}{i})B(\;\undersym{X}{i})-C([\;\undersym{X}{i},\;\undersym{X}{j}],\;\undersym{X}{j},\;\undersym{X}{i})\\
         & &-\Big(C(\;\undersym{X}{j},\,\oversetc{\nabla}_{\;\undersym{X}{j}}\;\undersym{X}{i},\;\undersym{X}{i}) +C(\;\undersym{X}{j},\,[\;\undersym{X}{i},\;\undersym{X}{j}],\;\undersym{X}{i})\Big)\\
         &=&-2\;\undersym{X}{i}\cdot B(\;\undersym{X}{i})+B(\;\undersym{X}{i})B(\;\undersym{X}{i})+C\big(T(\;\undersym{X}{i},\;\undersym{X}{j}),\;\undersym{X}{j},\;\undersym{X}{i}\big)- C\big(\;\undersym{X}{j},\,\oversetc{\nabla}_{\;\undersym{X}{i}}\;\undersym{X}{j},\;\undersym{X}{i}\big)\\
         &=&-2\;\undersym{X}{i}\cdot B(\;\undersym{X}{i})+B(\;\undersym{X}{i})B(\;\undersym{X}{i})+C\big(T(\;\undersym{X}{i},\;\undersym{X}{j}),\;\undersym{X}{j},\;\undersym{X}{i}\big)+ C\big(\;\undersym{X}{j},\,C(\;\undersym{X}{i},\;\undersym{X}{j}),\;\undersym{X}{i}\big).
  \end{eqnarray*}
\end{description}
\vspace*{-1.2cm}\[\qedhere\]
\end{proof}

\vspace*{.26cm}

\begin{center}
    \large{Table1: Linear connections in AP-geometry}
\end{center}
\begin{center}
    \begin{tabular}{|c|c|c|c|c|}
       \hline
       % after \\: \hline or \cline{col1-col2} \cline{col3-col4} ...
       \multirow{2}{*}{Connection} & \multirow{2}{*}{Symbol}     & \multirow{2}{*}{Torsion} & \multirow{2}{*}{Curvature} & \multirow{2}{*}{Metricity} \\
       &&&&\\ \hline
       \multirow{2}{*}{Canonical}  &     \multirow{2}{*}{$\nabla$}  &   \multirow{2}{*}{$T$}   &    \multirow{2}{*}{$0$}    & \multirow{2}{*}{metric}     \\
       &&&&\\ \hline
       \multirow{2}{*}{Dual}       & \multirow{2}{*}{$\widetilde{\nabla}$} &   \multirow{2}{*}{$-T$}  & \multirow{2}{*}{$\widetilde{R}$}  & \multirow{2}{*}{nonmetric}  \\
       &&&&\\ \hline
       \multirow{2}{*}{Symmetric}  & \multirow{2}{*}{$\widehat{\nabla}$} &    \multirow{2}{*}{$0$}  & \multirow{2}{*}{$\widehat{R}$}  & \multirow{2}{*}{nonmetric}  \\
       &&&&\\ \hline
       \multirow{2}{*}{Levi-Civita}& \multirow{2}{*}{$\oversetc{\nabla}$} &    \multirow{2}{*}{$0$}  & \multirow{2}{*}{$\oversetc{R}$}  & \multirow{2}{*}{metric}     \\
       &&&&\\ \hline
     \end{tabular}
\end{center}

\vspace*{1cm}

Let $D$ be an arbitrary linear connection on $M$ with torsion ${\bf{T}}$ and curvature ${\bf{R}}$. Then the Bianchi identities are given, for all $Y,\,Z,\,V,\,U\in\mathfrak{X}(M)$, by \cite{Kobayashi}:
\begin{description}
  \item[First Bianchi identity:] $\underset{Y,\,Z,\,V}{\;\mathfrak{S}}\Big\{\;{\bf{R}}(Y,\,Z)V\Big\}=\underset{Y,\,Z,\,V}{\mathfrak{S}}\;\Big\{(D_V{\bf{T}})(Y,\,Z) +{\bf{T}}\big({\bf{T}}(Y,\,Z),\,V\big)\Big\}.$
  \item[Second Bianchi identity:] $\underset{Y,\,Z,\,V}{\mathfrak{S}}\;\Big\{(D_V {\bf{R}})(Y,\,Z)U-{\bf{R}}\big(V,\,{\bf{T}}(Y,\,Z)\big)U\Big\}=0$.
\end{description}

In what follows, we derive some identities using the above Bianchi identities. Some of the derived identities will be used to simplify other formulae thus obtained.

\begin{proposition} \label{1st}
The first Bianchi identity for the connections $\nabla,\,\widetilde{\nabla},\,\widehat{\nabla}$ and $\oversetc{\nabla}$ reads:
\begin{description}
  \item[(a)] $\underset{Y,\,Z,\,V}{\mathfrak{S}}\;\Big\{(\nabla_VT)(Y,\,Z)+T\big(T(Y,\,Z),\,V\big)\Big\}=0.$ %\hfill\refstepcounter{equation}(\theequation)\label{1st-canonical}$
  \item[(b)] $\underset{Y,\,Z,\,V}{\mathfrak{S}}\;\Big\{\widetilde{R}(Y,\,Z)V\Big\}=\underset{Y,\,Z,\,V}{\mathfrak{S}}\;\Big\{T\big(T(Y,\,Z),\,V\big) -(\widetilde{\nabla}_VT)(Y,\,Z)\Big\}.$ %\hfill\refstepcounter{equation}(\theequation)\label{1st-dual}$
  \item[(c)] $\underset{Y,\,Z,\,V}{\mathfrak{S}}\;\Big\{\widehat{R}(Y,\,Z)V\Big\}=0.$%\hfill\refstepcounter{equation}(\theequation)\label{1st-sym}$
  \item[(d)] $\underset{Y,\,Z,\,V}{\mathfrak{S}}\;\Big\{\;\overcirc{R}(Y,\,Z)V\Big\}=0.$%\hfill\refstepcounter{equation}(\theequation)\label{1st-riemannian}$
\end{description}
\end{proposition}

The proof is straightforward. We have to use the relations $R=0,\,\widetilde{T}=-T$ and $\widehat{T}=\oversetc{T}=0$.

\begin{corollary}\label{corrbian}
The following identities hold:
\begin{description}
  \item[(a)] $\underset{Y,\,Z,\,V}{\mathfrak{S}}\;\Big\{(\widetilde{\nabla}_VT)(Y,\,Z)\Big\} =2\;\underset{Y,\,Z,\,V}{\mathfrak{S}}\Big\{T\big(T(Y,\,Z),\,V\big)\Big\}.$%\hfill\refstepcounter{equation}(\theequation)\label{1st1st}$
  \item[(b)] $\underset{Y,\,Z,\,V}{\mathfrak{S}}\;\Big\{(\widehat{\nabla}_VT)(Y,\,Z)\Big\} =\frac{1}{2}\;\underset{Y,\,Z,\,V}{\mathfrak{S}}\Big\{T\big(T(Y,\,Z),\,V\big)\Big\}.$%\hfill\refstepcounter{equation}(\theequation)\label{2nd2nd}$
  \item[(c)] $\underset{Y,\,Z,\,V}{\mathfrak{S}}\;\Big\{\widetilde{R}(Y,\,Z)V\Big\} =-\underset{Y,\,Z,\,V}{\mathfrak{S}}\Big\{T\big(T(Y,\,Z),\,V\big)\Big\}.$%\hfill\refstepcounter{equation}(\theequation)\label{3rd3rd}$
\end{description}
\end{corollary}

The proof follows from the above proposition together with Corollary \ref{difofconnections} and (\ref{dual}).

\begin{proposition} \label{2nd}
The second Bianchi identity for the connections $\widetilde{\nabla},\,\widehat{\nabla}$ and $\oversetc{\nabla}$ reads:
\begin{description}
  \item[(a)] $\underset{Y,\,Z,\,V}{\mathfrak{S}}\;\Big\{(\widetilde{\nabla}_V\widetilde{R})(Y,\,Z)U\Big\} =\underset{Y,\,Z,\,V}{\mathfrak{S}}\;\Big\{(\nabla_UT)\big(T(Y,\,Z),\,V\big)\Big\}.$ %\hfill\refstepcounter{equation}(\theequation)\label{2nd-dual}$
  \item[(b)] $\underset{Y,\,Z,\,V}{\mathfrak{S}}\;\Big\{(\widehat{\nabla}_V\widehat{R})(Y,\,Z)U\Big\} =0.$%\hfill\refstepcounter{equation}(\theequation)\label{2nd-sym}$
  \item[(c)] $\underset{Y,\,Z,\,V}{\mathfrak{S}}\;\Big\{(\;\overcirc{\nabla}_V\oversetc{R})(Y,\,Z)U\Big\} =0.$%\hfill\refstepcounter{equation}(\theequation)\label{2nd-riemannian}$
\end{description}
\end{proposition}

The proof is straightforward making use of (\ref{dual}).

\vspace*{.26cm}Now, we will give another formula for the curvature tensor $\widehat{R}$ of the symmetric connection $\widehat{\nabla}$ which is more compact than (\ref{sym}).

\begin{theorem}
The curvature tensor $\widehat{R}$ can be written in the form:
\begin{equation*} %\label{sym2}
\widehat{R}(Y,\,Z)V=\frac{1}{2}(\nabla_VT)(Y,\,Z)-\frac{1}{4}\Big(T\big(Y,\,T(Z,\,V)\big)+T\big(Z,\,T(V,\,Y)\big)\Big).
\end{equation*}
\end{theorem}

\begin{proof}
Taking into account (\ref{sym}) and Proposition \ref{1st}{\bf{(a)}}, one has
\begin{eqnarray*}
\widehat{R}(Y,\,Z)V&=&-\frac{1}{2}\;\underset{Y,\,Z,\,V}{\mathfrak{S}}\;\Big\{(\nabla_VT)(Y,\,Z)\Big\}+\frac{1}{2}(\nabla_VT)(Y,\,Z) \\
              & &+\frac{1}{4}\;\underset{Y,\,Z,\,V}{\mathfrak{S}}\;\Big\{T\big(V,\,T(Y,\,Z)\big)\Big\}+\frac{1}{4}T\big(V,\,T(Y,\,Z)\big)\\
              &=&-\frac{1}{4}\;\underset{Y,\,Z,\,V}{\mathfrak{S}}\;\Big\{T\big(V,\,T(Y,\,Z)\big)\Big\}+\frac{1}{4}T\big(V,\,T(Y,\,Z)\big) +\frac{1}{2}(\nabla_VT)(Y,\,Z)\\
              &=&\frac{1}{2}(\nabla_VT)(Y,\,Z)-\frac{1}{4}\Big(T\big(Y,\,T(Z,\,V)\big)+T\big(Z,\,T(V,\,Y)\big)\Big).
\end{eqnarray*}
\vspace*{-1.6cm}\[\qedhere\]
\end{proof}

\begin{corollary}\label{exception1}
On an AP-space $(M,\;\undersym{X}{i})$ the Ricci-like tensor ${\rm\widehat{R}ic}$ with respect to the symmetric connection $\widehat{\nabla}$ can be written as:
\begin{equation*} %\label{sym--ricci}
{\rm\widehat{R}ic}(Y,\,Z)=-\frac{1}{2}(\nabla_ZB)(Y)+\frac{1}{4}B\big(T(Y,\,Z)\big)-\frac{1}{4}T\big(Y,\,T(\;\undersym{X}{i},\,Z),\;\undersym{X}{i}\big).
\end{equation*}
\end{corollary}

\vspace*{.26cm}

It is to be noted that the expression $\underset{Y,\,Z,\,V}{\mathfrak{S}}\Big\{T\big(T(Y,\,Z),\,V\big)\Big\}$ appears in many of the identities obtained above. We discuss now the case in which this expression vanishes.

Let us write $[\;\undersym{X}{i},\;\undersym{X}{j}]=:C^h_{ij}\;\undersym{X}{h}$. The functions $C^h_{ij}\in\mathfrak{F}(M)$ are global functions on $M$ and will be referred to as the global structure coefficients of the AP-space. They can be written explicitly in the form $C^h_{ij}=\;\undersym{\Omega}{h}([\;\undersym{X}{i},\;\undersym{X}{j}])$. The last expression may be considered as a definition of the global structure coefficients.

\begin{theorem}\label{Golden}
On an AP-space $(M,\;\undersym{X}{i})$ the expression $\underset{Y,\,Z,\,V}{\mathfrak{S}}\Big\{T\big(T(Y,\,Z),\,V\big)\Big\}$ vanishes if and only if, for all $h$, the expression $\underset{i,\,j,\,k}{\mathfrak{S}}\Big\{\;\undersym{X}{k}\cdot C^h_{ij}\Big\}$ vanishes.\\
Consequently, if the global structure coefficients of the AP-space are constant functions on $M$, then $\underset{Y,\,Z,\,V}{\mathfrak{S}}\Big\{T\big(T(Y,\,Z),\,V\big)\Big\}=0$.
\end{theorem}

\begin{proof}
Using the parallelization vector fields instead of $Y,\,Z$ and $V$, we have:
\begin{eqnarray*}
\underset{i,\,j,\,k}{\mathfrak{S}}\Big\{T\big(T(\;\undersym{X}{i},\;\undersym{X}{j}),\;\undersym{X}{k}\big)\Big\}=0
&\Longleftrightarrow&-\underset{i,\,j,\,k}{\mathfrak{S}}\Big\{T([\;\undersym{X}{i},\;\undersym{X}{j}],\;\undersym{X}{k})\Big\}=0,\;\text{by (\ref{torsion})}\\
&\Longleftrightarrow&\underset{i,\,j,\,k}{\mathfrak{S}}\Big\{\nabla_{\;\undersym{X}{k}}[\;\undersym{X}{i},\;\undersym{X}{j}] +\big[[\;\undersym{X}{i},\;\undersym{X}{j}],\;\undersym{X}{k}\big]\Big\}=0,\;\text{by (\ref{AP-cond})}\\
&\Longleftrightarrow&\underset{i,\,j,\,k}{\mathfrak{S}}\Big\{\nabla_{\;\undersym{X}{k}}[\;\undersym{X}{i},\;\undersym{X}{j}]\Big\}=0,\;\text{by Jacobi identity}\\
&\Longleftrightarrow&\underset{i,\,j,\,k}{\mathfrak{S}}\Big\{\big(\;\undersym{X}{k}\cdot \;\undersym{\Omega}{h}([\;\undersym{X}{i},\;\undersym{X}{j}])\big)\;\undersym{X}{h}\Big\}=0,\;\text{by (\ref{canonical})}\\
&\Longleftrightarrow&\underset{i,\,j,\,k}{\mathfrak{S}}\Big\{\big(\;\undersym{X}{k}\cdot \;\undersym{\Omega}{h}([\;\undersym{X}{i},\;\undersym{X}{j}])\big)\Big\}\;\undersym{X}{h}=0\\
&\Longleftrightarrow&\underset{i,\,j,\,k}{\mathfrak{S}}\Big\{\big(\;\undersym{X}{k}\cdot \;\undersym{\Omega}{h}([\;\undersym{X}{i},\;\undersym{X}{j}])\big)\Big\}=0\;\,\forall h,\;\text{by the independence of $\;\undersym{X}{i}$}\\
&\Longleftrightarrow&\underset{i,\,j,\,k}{\mathfrak{S}}\Big\{\;\undersym{X}{k}\cdot C^h_{ij}\Big\}=0\;\,\forall h,\;\text{by (\ref{base})}.
\end{eqnarray*}
\vspace*{-1.6cm}\[\qedhere\]
\end{proof}

\vspace*{.1cm} It should be noted that for the natural basis
$\{\frac{\partial}{\partial x^{\alpha}}:\alpha=1,\,...,\,n\}$, the bracket
$[\frac{\partial}{\partial x^{\alpha}},\,\frac{\partial}{\partial x^{\beta}}]=0$ and so the
structure coefficients associated with $(\frac{\partial}{\partial x^{\alpha}})$
vanish. For this reason the \textit{local} expression (in the
natural basis) of the identity
$\underset{Y,\,Z,\,V}{\mathfrak{S}}\Big\{T\big(T(Y,\,Z),\,V\big)\Big\}=0$ is
valid as is established in \cite{local}.

\vspace*{.1cm}The last proposition gives rise to the following interesting formulae.

\begin{corollary}\label{golden}
In an AP-space $(M,\;\undersym{X}{i})$, if the global structure coefficient of the AP-space are constant functions on $M$, then the next formulae hold:
\begin{description}
  \item[(a)] $(\nabla_VT)(Y,\,Z)=(\widetilde{\nabla}_VT)(Y,\,Z)=(\widehat{\nabla}_VT)(Y,\,Z)$.
  \item[(b)] $\underset{Y,\,Z,\,V}{\mathfrak{S}}\Big\{(\nabla_VT)(Y,\,Z)\Big\}=0$.
  \item[(c)] $\underset{Y,\,Z,\,V}{\mathfrak{S}}\Big\{(\widetilde{\nabla}_VT)(Y,\,Z)\Big\}=0$.
  \item[(d)] $\underset{Y,\,Z,\,V}{\mathfrak{S}}\Big\{(\widehat{\nabla}_VT)(Y,\,Z)\Big\}=0$.
  \item[(e)] $\underset{Y,\,Z,\,V}{\mathfrak{S}}\Big\{\widetilde{R}(Y,\,Z)V\Big\}=0$.
  \item[(f)] $\widehat{R}(Y,\,Z)V=\frac{1}{2}(\nabla_VT)(Y,\,Z)-\frac{1}{4}T\big(T(Y,\,Z),\,V\big)$.
  \item[(h)] $\underset{Y,\,Z,\,V}{\mathfrak{S}}\Big\{(\nabla_VC)(Y,\,Z)\Big\}=\underset{Y,\,Z,\,V}{\mathfrak{S}}\Big\{(\nabla_VC)(Z,\,Y)\Big\}$.
\end{description}
\end{corollary}

%%%%%%%%%%%%%%%%%%%%%%%%%%%%%%%%%%%%%%%%%%%%%%%%%%%%%%%%%%%%%%% Section 4: Wanas Tensor %%%%%%%%%%%%%%%%%%%%%%%%%%%%%%%%%%%%%%%%%%%%%%%%%%%%%%%%%%%%%%%

\section{Wanas Tensor}

\hspace*{.4cm} The Wanas tensor was first defined locally by M. I. Wanas in 1975
\cite{unificationT}. It has been used by F. Mikhail and M. Wanas
\cite{unificationP} to construct a pure geometric theory unifying
gravity and electromagnetism. In this section, we introduce the
global definition of the Wanas tensor and investigate it.

\begin{definition}
Let $(M,\;\undersym{X}{i})$ be an AP-space. The tensor field $W$ of type (1,\,3) on $M$ defined by the formula
\begin{equation*} %\label{wanas0}
W(Y,\,Z)\;\undersym{X}{i}=\widetilde{\nabla}^2_{\;Y,\,Z}\;\undersym{X}{i}-\widetilde{\nabla}^2_{\;Z,\,Y}\;\undersym{X}{i},
\end{equation*}
where $\,\widetilde{\nabla}^2_{\;Y,\,Z}=\widetilde{\nabla}_Y\widetilde{\nabla}_Z-\widetilde{\nabla}_{\widetilde{\nabla}_YZ}$, is called the Wanas tensor, or simply the W-tensor, of $(M,\;\undersym{X}{i})$.
\end{definition}

Using (\ref{base}), for every $Y,\,Z,\,V\in\mathfrak{X}(M)$, we get
\begin{equation}\label{wanas1}
W(Y,\,Z)V=\big(\widetilde{\nabla}^2_{\;Y,\,Z}\;\undersym{X}{i}-\widetilde{\nabla}^2_{\;Z,\,Y}\;\undersym{X}{i}\big)\;\undersym{\Omega}{i}(V).
\end{equation}

The next result gives a quite simple expression for a such tensor.

\begin{theorem}
The W-tensor satisfies the following identity
\begin{equation} \label{impw}
W(Y,\,Z)V=\widetilde{R}(Y,\,Z)V-T\big(T(Y,\,Z),\,V\big).
\end{equation}
\end{theorem}

\begin{proof}
Consider the commutation formula for the parallelization vector field $\;\undersym{X}{i}$ with respect to $\widetilde{\nabla}$:
$$\widetilde{\nabla}^2_{\;Y,\,Z}\;\undersym{X}{i}-\widetilde{\nabla}^2_{\;Z,\,Y}\;\undersym{X}{i}=\widetilde{R}(Y,\,Z)\;\undersym{X}{i} -\widetilde{\nabla}_{\widetilde{T}(Y,\,Z)}\;\undersym{X}{i}$$
Consequently,
\begin{eqnarray*}
W(Y,\,Z)V&\overset{(\ref{wanas1})}{=}&\;\undersym{\Omega}{i}(V)\widetilde{R}(Y,\,Z)\;\undersym{X}{i} -\;\undersym{\Omega}{i}(V)\widetilde{\nabla}_{\widetilde{T}(Y,\,Z)}\;\undersym{X}{i} \\
         &=&\widetilde{R}(Y,\,Z)V+\;\undersym{\Omega}{i}(V)\widetilde{\nabla}_{T(Y,\,Z)}\;\undersym{X}{i},\;\text{by (\ref{base})} \\
         &=&\widetilde{R}(Y,\,Z)V+\widetilde{\nabla}_{T(Y,\,Z)}\;\undersym{\Omega}{i}(V)\;\undersym{X}{i} -\big(T(Y,\,Z)\cdot\;\undersym{\Omega}{i}(V)\big)\;\undersym{X}{i} \\
         &=&\widetilde{R}(Y,\,Z)V+\widetilde{\nabla}_{T(Y,\,Z)}V-\nabla_{T(Y,\,Z)}V,\;\text{by (\ref{base}) and (\ref{canonical})} \\
         &=&\widetilde{R}(Y,\,Z)V-T\big(T(Y,\,Z),\,V\big),\;\text{by Proposition \ref{different-connections}}.
\end{eqnarray*}
\vspace*{-1.6cm}\[\qedhere\]
\end{proof}

\begin{corollary}
The W-tensor can be expressed in the form:
\begin{equation}\label{wanas2}
W(Y,\,Z)V=(\nabla_VT)(Y,\,Z)-T\big(T(Y,\,Z),\,V\big).
\end{equation}
\end{corollary}

In fact, this expression follows from (\ref{dual}). This shows that the W-tensor is expressed in terms of the torsion tensor of the AP-space only.

\begin{proposition} \label{w1stt}
The Wanas tensor has the following properties:
\begin{description}
  \item[(a)] $W(Y,\,Z)V$ is skew symmetric in the first two arguments $Y,\,Z$.
  \item[(b)] $\underset{Y,\,Z,\,V}{\mathfrak{S}}\;\Big\{W(Y,\,Z)\,V\Big\} =-2\;\underset{Y,\,Z,\,V}{\mathfrak{S}}\;\Big\{T\big(T(Y,\,Z),\,V\big)\Big\}.$%\hfill\refstepcounter{equation}(\theequation)\label{w1st}$
  \item[(c)] $\underset{Y,\,Z,\,V}{\mathfrak{S}}\;\Big\{W(Y,\,Z)\,V\Big\} =-\;\underset{Y,\,Z,\,V}{\mathfrak{S}}\;\Big\{(\widetilde{\nabla}_VT)(Y,\,Z)\Big\}.$%\hfill\refstepcounter{equation}(\theequation)\label{w1st1st}$
\end{description}
\end{proposition}

\begin{proof}
Property {\bf{(a)}} is trivial, {\bf{(b)}} follows from Proposition \ref{1st}{\bf{(a)}} and (\ref{wanas2}), {\bf{(c)}} follows from {\bf{(b)}} and Corollary \ref{corrbian}{\bf{(a)}}.
\end{proof}

{The identity satisfied by the W-tensor in Proposition \ref{w1stt}{\bf{(b)}} is the same as the first Bianchi identity \big(Corollary \ref{corrbian}{\bf{(c)}}\big) of the dual curvature tensor up to a constant. The
identity corresponding to the second Bianchi identity is given by:}
\begin{proposition}
The W-tensor satisfies the following identity:
\begin{eqnarray} \label{w2nd}
& & \underset{V,\,Y,\,Z}{\mathfrak{S}}\;\Big\{(\widetilde{\nabla}_VW)(Y,\,Z)U\Big\} \nonumber \\ &=&-\underset{V,\,Y,\,Z}{\mathfrak{S}}\;\Big\{T\Big(T\big(T(Y,\,Z),\,V\big),\,U\Big)+T\Big(T\big(T(Y,\,Z),\,U\big),\,V\Big) +T\big(T(U,\,V),\,T(Y,\,Z)\big)\Big\} \nonumber \\
& &+\underset{V,\,Y,\,Z}{\mathfrak{S}}\;\Big\{(\nabla_UT)\big(T(Y,\,Z),\,V\big)-(\nabla_VT)\big(T(Y,\,Z),\,U\big)\Big\}
\end{eqnarray}
\end{proposition}
\begin{proof}
Taking into account (\ref{impw}) together with Proposition \ref{2nd}{\bf{(a)}}, Corollary \ref{corrbian}{\bf{(a)}} and Corollary \ref{difofconnections}{\bf{(a)}}, we get
\begin{eqnarray*}
& &\underset{V,\,Y,\,Z}{\mathfrak{S}}\;\Big\{(\widetilde{\nabla}_VW)(Y,\,Z)U\Big\} \\
&=&\underset{V,\,Y,\,Z}{\mathfrak{S}}\;\Big\{\Big((\widetilde{\nabla}_V\widetilde{R})(Y,\,Z)U -(\widetilde{\nabla}_VT)\big(T(Y,\,Z),\,U\big)-T\big((\widetilde{\nabla}_VT)(Y,\,Z),\,U\big)\Big)\Big\}\\
&=&\underset{V,\,Y,\,Z}{\mathfrak{S}}\;\Big\{(\nabla_UT)\big(T(Y,\,Z),\,V\big)-(\widetilde{\nabla}_VT)\big(T(Y,\,Z),\,U\big) -2T\Big(T\big(T(Y,\,Z),\,V\big),\,U\Big)\Big\} \\
&=&\underset{V,\,Y,\,Z}{\mathfrak{S}}\;\Big\{(\nabla_UT)\big(T(Y,\,Z),\,V\big)-(\nabla_VT)\big(T(Y,\,Z),\,U\big)\Big\} \\
& &-\underset{V,\,Y,\,Z}{\mathfrak{S}}\;\Big\{2T\Big(T\big(T(Y,\,Z),\,V\big),\,U\Big) +\underset{V,\,T(Y,\,Z),\,U}{\mathfrak{S}}\;T\Big(T\big(T(Y,\,Z),\,U\big),\,V\Big)\Big\} \\
&=&\underset{V,\,Y,\,Z}{\mathfrak{S}}\;\Big\{(\nabla_UT)\big(T(Y,\,Z),\,V\big)-(\nabla_VT)\big(T(Y,\,Z),\,U\big)\Big\} \\
& &-\underset{V,\,Y,\,Z}{\mathfrak{S}}\;\Big\{T\Big(T\big(T(Y,\,Z),\,V\big),\,U\Big)+T\Big(T\big(T(Y,\,Z),\,U\big),\,V\Big) +T\big(T(U,\,V),\,T(Y,\,Z)\big)\Big\}.
\end{eqnarray*}
\vspace*{-1.6cm}\[\qedhere\]
\end{proof}

\begin{corollary}
In an AP-space $(M,\;\undersym{X}{i})$, if the global structure coefficient of the AP-space are constant, we have
\begin{description}
  \item[(a)] $\underset{V,\,Y,\,Z}{\mathfrak{S}}\;\Big\{W(Y,\,Z)V\Big\}=0$.
  \item[(b)] $\underset{V,\,Y,\,Z}{\mathfrak{S}}\;\Big\{(\widetilde{\nabla}_VW)(Y,\,Z)U\Big\} =\underset{V,\,Y,\,Z}{\mathfrak{S}}\;\Big\{(\nabla_UT)\big(T(Y,\,Z),\,V\big)- (\nabla_VT)\big(T(Y,\,Z),\,U\big)\Big\}$
\end{description}
\end{corollary}

The proof is straightforward from Theorem \ref{Golden} and Corollary \ref{golden}.

\vspace{8pt} We end this section by the following comments and
remarks on Wanas tensor.
\begin{itemize}
      \item The W-tensor is defined by using the commutation formula with respect to the dual connection $\widetilde{\nabla}$.
      Nothing new arose from the same definition if we use the three other connections ($\nabla,\,\widehat{\nabla}$ and $\oversetc{\nabla}$).
      \item Using the commutation formula for the parallelization form $\;\undersym{\Omega}{i}$ instead of the parallelization
      vector field $\;\undersym{X}{i}$ in the definition of the
      W-tensor:
      \begin{equation*} %\label{wanas3}
      W(Y,\,Z)V=\Big((\widetilde{\nabla}^2_{\;Z,\,Y}\;\undersym{\Omega}{i})(V) -(\widetilde{\nabla}^2_{\;Y,\,Z}\;\undersym{\Omega}{i})(V)
      \Big)\;\undersym{X}{i}
      \end{equation*}
      gives the same formula (\ref{impw}) for the W-tensor and consequently the same properties.
      \item Being defined by using the parallelization vector fields $\;\undersym{X}{i}$, the Wanas tensor is defined
      only in AP-geometry. It has no analogue in other geometries.
      \item Although the W-tensor and the dual curvature tensor have some common properties (for example,
      Proposition \ref{w1stt}{\bf{(b)}}), there are significantly different properties (for example, (\ref{w2nd})).
      In the case of constant global structure coefficients, the W-tensor has some properties common with the Riemannian curvature $\;\overcirc{R}$.
      \item For a physical discussion concerning the W-tensor we refer to \cite{local}.
    \end{itemize}

%%%%%%%%%%%%%%%%%%%%%%%%%%%%%%%%%%%%%%%%%%%%%%%%%%%%% Section 4: Parallelization versus natural basis %%%%%%%%%%%%%%%%%%%%%%%%%%%%%%%%%%%%%%%%%%%%%%%%%%%

\section{Parallelization basis versus natural basis}

\hspace*{.4cm} This section is devoted to a double-view for the fundamental
geometric objects of AP-geometry. On one hand, we consider the local
expressions of these geometric objects in the natural basis
\cite{local} and, on the other hand, we compute their expressions in
the parallelization basis, giving rise to a concise table expressing
this double-view.

\vspace*{.26cm}Let $\big(U,\,(x^{\alpha})\big)$ be a local
coordinate system of $M$. At each point $x\in U$, we have two
distinguished bases of $T_xM$, namely, the natural basis
$\{\partial_{\mu}:=\frac{\partial}{\partial x^{\mu}}:
\mu=1,\,...,\,n\}$ and the parallelization basis
$\{\;\undersym{X}{i}(x):i=1,\,...,\,n\}$. These two bases are
fundamentally different. The parallelization vector fields
$\;\undersym{X}{i}$ are defined globally on the manifold $M$ whereas
the natural basis vector fields $\partial_{\mu}$ are defined only on
the coordinate neighborhood $U$. Consequently, the natural basis
vector fields depend crucially on coordinate systems whereas the
parallelization vector fields do not.

\vspace*{.26cm} Greek (world) indices are related to the natural
basis and Latin (mesh) indices are related to the parallelization
basis. Einstein summation convention will be applied as usual on
Greek indices. It will also be applied on Latin indices whatever
their position is (even if the two repeated indices are upward or
downward).

\vspace*{.26cm}A tensor field $H$ of type $(r,\,s)$ on $M$ is
written in the natural basis in the form:
$$\;H=H^{\alpha_1\,...\,\alpha_r\,}_{{\mu_1\,...\,\mu_s\,}}\partial_{\alpha_1}\otimes...
\otimes\partial_{\alpha_r}\otimes{dx^{\mu_1}}\otimes...\otimes{dx^{\mu_s}},\,\,
\text{on $U$}$$ and in the parallelization basis in the form:
$$\;H=H^{i_1\,...\,i_r\,}_{{j_1\,...\,j_s\,}}\;\underset{i_1}{X}
\otimes...\otimes\;\underset{i_r}{X}\otimes\;\underset{j_1}{\Omega}\otimes...\otimes\;\underset{j_s}{\Omega},
\,\, \text{on}\,\, M,$$ where\,
$H^{\alpha_1\,...\,\alpha_r}_{\mu_1\,...\,\mu_s}\in\mathfrak{F}(U)$
and $H^{i_1\,...\,i_r}_{j_1\,...\,j_s}\in\mathfrak{F}(M)$.

\vspace*{.26cm}A vector field $Y\in\mathfrak{X}(M)$ is written in the natural basis
in the form $Y=Y^{\alpha}\partial_{\alpha}$ and in the
parallelization basis in the form $Y=Y^{i}\;\undersym{X}{i}$. In
particular,
$\;\undersym{X}{i}=\;\undersym{X}{i}^{\alpha}\partial_{\alpha}$\,
and \,
$\partial_{\alpha}=\;\undersym{\Omega}{i}(\partial_{\alpha})\;\undersym{X}{i}=\;\undersym{\Omega}{i}_{\alpha}\;\undersym{X}{i}$.
Hence\, $(\;\undersym{X}{i}^{\alpha})_{1\leq\alpha,i\leq n}$\, is
the matrix of change of bases and\,
$(\;\undersym{\Omega}{i}_{\alpha})_{1\leq\alpha,i\leq n}$\, is the
inverse matrix.

\vspace*{.26cm}We use the following notations (with similar notations with respect to mesh indices):\\
$\Gamma^{\alpha}_{\mu\nu},\,\widetilde{\Gamma}^{\alpha}_{\mu\nu},\,\widehat{\Gamma}^{\alpha}_{\mu\nu},\;\overcirc{\Gamma}^{\alpha}_{\mu\nu}$: the coefficients
of the linear connections $\nabla,\,\widetilde{\nabla},\,\widehat{\nabla},\,\oversetc{\nabla}$ respectively,\\
$\widetilde{|}$: the covariant derivative with respect to the dual connection $\widetilde{\nabla}$,\\
$g_{\mu\nu}$ (resp. $g^{\mu\nu}$): the covariant (resp. contravariant) components of the metric tensor $g$,\\
$\Lambda^{\alpha}_{\mu\nu}$: the components of the torsion tensor $T$,\\
$B_{\alpha}$: the components of the basic form $B$,\\
$\gamma^{\alpha}_{\mu\nu}$: the components of the contortion tensor $C$,\\
$W^{\alpha}_{\sigma\mu\nu}$: the components of the Wanas tensor $W$.

\vspace*{.26cm}Let $D$ be an arbitrary connection on $M$ with
torsion tensor $T$ and curvature tensor $R$. We use the following
conventions:
$D_{\partial_{\mu}}\partial_{\nu}=D^{\alpha}_{\nu\mu}\,\partial_{\alpha}$,\,
$T(\partial_{\mu},\,\partial_{\nu})=T^{\alpha}_{\nu\mu}\,\partial_{\alpha}$,\,
$R(\partial_{\mu},\,\partial_{\nu})\partial_{\sigma}=R^{\alpha}_{\sigma\mu\nu}\partial_{\alpha}$,
with similar conventions with respect to Latin indices.

\vspace*{.26cm} The next table gives a comparison between the most
important geometric objects of AP-geometry expressed in the natural
basis and in the parallelization basis. Geometric objects, equations
or identities having the same form in the two bases are not
included in that table. However, if a
geometric object has the same form in the two bases, that is, if its
expressions in world indices and mesh indices are similar, this
does not mean that the geometric meaning of these two expressions is
the same.

\clearpage

\begin{center}
    \large{Table2: Parallelization basis versus natural basis}
\end{center}

\begin{center}
\begin{tabular}{|c|c|c|}
      \hline
      % after \\: \hline or \cline{col1-col2} \cline{col3-co-l4} ...
   \multirow{3}{*}{Geometric object}& Local form             & Global form                      \\
                                    & In the natural basis    & In the parallelization basis      \\
                                    & (world indices)        & (mesh indices)                  \\ \hline
   &\multirow{3}{*}{$\;\undersym{X}{i}^{\alpha}\;\undersym{\Omega}{j}_{\alpha}=\delta_{ij}, \quad\;\undersym{X}{i}^{\alpha}\;\undersym{\Omega}{i}_{\mu}=\delta^{\alpha}_{\mu}$}&\multirow{3}{*}{$\;\undersym{X}{j}^{k}=\delta_{j}^{k}, \quad\;\undersym{\Omega}{j}_{k}=\delta_{jk}$} \\ Parallelization vector fields, & & \\
   parallelization forms & & \\[8pt] \hline
   \multirow{2}{*}{Metric tensor}&\multirow{2}{*}{$g_{\mu\nu}=\;\undersym{\Omega}{i}_{\mu}\;\undersym{\Omega}{i}_{\nu}$}&\multirow{2}{*}{$g_{jk}=\delta_{jk}$}\\
                                 &                                                                    &                                  \\ \hline
   \multirow{4}{*}{Canonical connection} & \multirow{2}{*}{$\Gamma^{\alpha}_{\nu\mu}=\;\undersym{X}{i}^{\alpha}\;\undersym{\Omega}{i}_{\nu,\mu}$} & \multirow{4}{*}{$\Gamma_{jk}^{h}=0$} \\ & & \\ & \multirow{2}{*}{where $,_{\mu}$ denotes $\partial_{\mu}$} & \\ & & \\ \hline
   \multirow{2}{*}{Dual connection}&\multirow{2}{*}{$\widetilde{\Gamma}^{\alpha}_{\nu\mu}=\Gamma^{\alpha}_{\mu\nu}$}&\multirow{2}{*}{$\widetilde{\Gamma}^h_{jk}=C^h_{kj}$}\\
                                   &                                                             &                                              \\ \hline
   \multirow{2}{*}{Symmetric connection}&\multirow{2}{*}{$\widehat{\Gamma}^{\alpha}_{\nu\mu}=\frac{1}{2}(\Gamma^{\alpha}_{\nu\mu}+\Gamma^{\alpha}_{\mu\nu})$}& \multirow{2}{*}{$\widehat{\Gamma}^h_{jk}= \frac{1}{2}C^h_{kj}$}\\ && \\ \hline
   \multirow{2}{*}{Levi-Civita connection}&\multirow{2}{*}{$\;\overcirc{\,\Gamma}^{\alpha}_{\nu\mu}=\frac{1}{2}g^{\alpha\sigma}(g_{\sigma\nu,\mu}+g_{\mu\sigma,\nu}- g_{\nu\mu,\sigma})$}&\multirow{2}{*}{$\;\overcirc{\,\Gamma}^h_{jk}=\frac{1}{2}\Big(C^h_{kj}+C^j_{hk}+C^k_{hj}\Big)$}\\ && \\ \hline
   \multirow{2}{*}{Torsion tensor}&\multirow{2}{*}{$\Lambda^{\alpha}_{\nu\mu}=\Gamma^{\alpha}_{\nu\mu}-\Gamma^{\alpha}_{\mu\nu}$}&\multirow{2}{*}{$\Lambda_{jk}^h=C^h_{jk}$} \\
                                  &                                                                           &                                       \\ \hline
  \multirow{2}{*}{Contortion tensor}&\multirow{2}{*}{$\gamma^{\alpha}_{\nu\mu}=\Gamma^{\alpha}_{\nu\mu}-\;\overcirc{\,\Gamma}^{\alpha}_{\nu\mu}$}& \multirow{2}{*}{$\gamma^h_{jk}=-\;\overcirc{\,\Gamma}^h_{jk}$}\\
         &                                                                                     &                                                   \\ \hline
   &\multirow{2}{*}{$\Lambda^{\alpha}_{\nu\mu}=\gamma^{\alpha}_{\nu\mu}-\gamma^{\alpha}_{\mu\nu}$}&\multirow{6}{*}{$\Lambda^{h}_{jk}=\gamma^{h}_{jk} -\gamma^{h}_{kj}$} \\ [-.15cm] & & \\
   Torsion in terms &\multirow{2}{*}{$\Lambda_{\sigma\nu\mu}=\gamma_{\sigma\nu\mu}-\gamma_{\sigma\mu\nu}$}& \\ of contortion & & \\[-3pt]
   & \multirow{2}{*}{where $\Lambda_{\mu\nu\sigma}=g_{\epsilon\mu}\Lambda^{\epsilon}_{\nu\sigma}$ and $\gamma_{\mu\nu\sigma}=g_{\epsilon\mu}\gamma^{\epsilon}_{\nu\sigma}$} & \\  & & \\ \hline
   &\multirow{2}{*}{$\gamma^{\alpha}_{\nu\mu}=\frac{1}{2}\Big(\Lambda^{\alpha}_{\nu\mu}+(\Lambda_{\mu\nu\epsilon} +\Lambda_{\nu\mu\epsilon})g^{\alpha\epsilon}\Big)$}&\multirow{4}{*}{${\gamma^h_{jk}}= \frac{1}{2}(C^h_{jk}+C^k_{jh}+C^j_{kh})$}\\
   Contortion in terms& &\\
   of torsion &\multirow{2}{*}{$\gamma_{\mu\nu\sigma}=\frac{1}{2}(\Lambda_{\sigma\nu\mu}+\Lambda_{\mu\nu\sigma}+\Lambda_{\nu\sigma\mu})$}&\\
   & &\\ \hline
   \multirow{2}{*}{Basic form}&\multirow{2}{*}{$B_{\mu}=\Lambda^{\alpha}_{\mu\alpha}=\gamma^{\alpha}_{\mu\alpha}$}&\multirow{2}{*}{$B_j=\Lambda^k_{jk}=\gamma^k_{jk}=C^k_{jk}$}\\ && \\ \hline
   \multirow{4}{*}{Wanas tensor}&\multirow{2}{*}{$W^{\alpha}_{\sigma\mu\nu}=(\;\undersym{X}{i}^{\alpha}_{\;\;\widetilde{|}\nu\mu} -\;\undersym{X}{i}^{\alpha}_{\;\;\widetilde{|}\mu\nu})\;\undersym{\Omega}{i}_{\sigma}$} &\multirow{2}{*}{$W^h_{kij}=\;\undersym{X}{k}^h_{\;\;\widetilde{|}ji}-\;\undersym{X}{k}^h_{\;\;\widetilde{|}ij}$}\\
   & & \\
   &\multirow{2}{*}{$W^{\alpha}_{\sigma\mu\nu}=\Lambda^{\epsilon}_{\mu\nu}\Lambda^{\alpha}_{\sigma\epsilon}-\Lambda^{\alpha}_{\mu\nu|\sigma}$}
   &\multirow{2}{*}{$W^h_{kij}=C^l_{ij}C^h_{kl}-\;\undersym{X}{k}\cdot C^h_{ij}$}\\
   & & \\ \hline
\end{tabular}
\end{center}

%\clearpage

The above table merits some comments. We conclude this section and
the paper by the following remarks and comments.
\begin{itemize}
\item The third column of the above table is obtained by computing
      the expression of the geometric objects in the parallelization
      basis. For example, to compute the coefficients of the Levi-Civita connection $\;\overcirc{\Gamma}^{h}_{jk}$, set\, $Y=\;\undersym{X}{k},\,
      Z=\;\undersym{X}{j},\, V=\;\undersym{X}{h}$ in (\ref{riemannian}).
      Then, we get
      \begin{eqnarray*}
      2g(\nabla_{\;\undersym{X}{k}}\;\undersym{X}{j},\;\undersym{X}{h})&=&\;\undersym{X}{k}\cdot g(\;\undersym{X}{j},\;\undersym{X}{h})
      +\;\undersym{X}{j}\cdot g(\;\undersym{X}{h},\;\undersym{X}{k})-\;\undersym{X}{h}\cdot g(\;\undersym{X}{k},\;\undersym{X}{j})\\
      & &-g(\;\undersym{X}{k},[\;\undersym{X}{j},\;\undersym{X}{h}])+g(\;\undersym{X}{j},[\;\undersym{X}{h},\;\undersym{X}{k}])
      +g(\;\undersym{X}{h},[\;\undersym{X}{k},\;\undersym{X}{j}]).
      \end{eqnarray*}
      For the left-hand side (LHS),
      $$LHS=2\, g(\;\overcirc{\Gamma}^l_{jk}\;\undersym{X}{l},\;\undersym{X}{h})=2\, g_{lh}\; \overcirc{\Gamma}^l_{jk}=
      2\, \delta_{lh}\; \overcirc{\Gamma}^l_{jk}=2\; \overcirc{\Gamma}^h_{jk}.$$
      As $\;\undersym{X}{k}\cdot g(\;\undersym{X}{j},\;\undersym{X}{h})=\;\undersym{X}{k}\cdot g_{jh}=\;\undersym{X}{k}\cdot \delta_{jh}=0$,
      the first three terms of the right-hand side (RHS) vanish.
      Hence,
      \vspace{-6pt}
      \begin{eqnarray*}
      % \nonumber to remove numbering (before each equation)
        RHS &=& -g(\;\undersym{X}{k},C^l_{jh}\;\undersym{X}{l})+g(\;\undersym{X}{j},C^l_{hk}\;\undersym{X}{l})+g(\;\undersym{X}{h},C^l_{kj}\; \undersym{X}{l})\\
            &=& -g_{kl}\, C^l_{jh}+g_{jl}\, C^l_{hk}+g_{hl}\, C^l_{kj} \\
            &=& -C^k_{jh}+C^j_{hk}+C^h_{kj}.
      \end{eqnarray*}
      Accordingly, $$\;\overcirc{\Gamma}^h_{jk}=\frac{1}{2}(C^j_{hk}+C^h_{kj}-C^k_{jh}).$$
\item It is clear from the third column that almost all geometric objects of AP-geometry are expressed in terms of the
      global structure coefficients $C^h_{jk}$.
      The global structure coefficients thus play a dominant role in AP-geometry
      formulated in mesh indices. Its role is similar to, and even more
      important than, the role played by the torsion tensor
      $\Lambda^h_{jk}$ in AP-geometry formulated in world indices.
\item The structure coefficients $C^\alpha_{\mu\nu}$ with respect to an arbitrary
      basis $(e_\alpha)$ are not the components of a $(1,2)$-tensor field. In
      fact, let $e_{\alpha'}=A^\alpha_{\alpha'}\,e_\alpha$ under a change of local
      coordinates from $(x^\alpha)$ to $(x^{\alpha'})$ and let
      $[e_\mu,e_\nu]=C^\alpha_{\mu\nu}e_\alpha$ and
      $[e_{\mu'},\,e_{\nu'}]=C^{\alpha'}_{\mu'\nu'}\, e_{\alpha'}$. Then, one can easily show that the transformation formula
      for $C^\alpha_{\mu\nu}$ has the form:
      $$C^{\alpha'}_{\mu'\nu'}=A^{\alpha'}_{\alpha}\, A^{\mu}_{\mu'}\,
      A^{\nu}_{\nu'}\,
      C^\alpha_{\mu\nu}+K^{\alpha'}_{\mu'\nu'}-K^{\alpha'}_{\nu'\mu'},$$
      where $K^\alpha_{\mu\nu}=A^{\mu'}_{\mu}\, A^{\alpha}_{\alpha'}\,(e_{\mu'}\cdot A^{\alpha'}_\nu)$.
      Thus, $C^\alpha_{\mu\nu}$ are not the components of a tensor field of
      type $(1,2)$ unless $e_{\mu'}\cdot A^{\alpha'}_{\nu}=0$ (that is, the matrix of
      change of bases ${A^{\alpha'}}_{\alpha}$ is a constant matrix) or $K^\alpha_{\mu\nu}$ is
      symmetric with respect to $\mu$ and $\nu$. Also, the global
      structure coefficients $C^h_{jk}$ are not the components of a
      $(1,2)$-tensor field (they are $n^3$ functions defined globally on $M$
      and having certain properties). Nevertheless, for fixed $j$ and
      $k$, $C^h_{jk}$ are the components of the $(1,0)$-tensor field $[\;\undersym{X}{j},\;\undersym{X}{k}]$.
\item In the parallelization basis, although the coefficients of the canonical connection $\nabla$ vanish: $\Gamma^h_{jk}=0$
      ($\nabla_{\;\undersym{X}{j}}\;\undersym{X}{k}=0$ because of the AP-condition), its
      torsion tensor $T$ does not vanish: $\Lambda^{h}_{jk}=C^h_{jk}$; a phenomenon that
      never exist in natural local coordinates. This is due to the
      non-vanishing of the bracket $[\;\undersym{X}{j},\;\undersym{X}{k}]$
      in the expression of the torsion tensor:
      $T(\;\undersym{X}{j},\;\undersym{X}{k})=\nabla_{\;\undersym{X}{j}}\;\undersym{X}{k}-\nabla_{\;\undersym{X}{k}}\;\undersym{X}{j}-
      [\;\undersym{X}{j},\;\undersym{X}{k}]$.
      For the same reason the dual connection $\widetilde{\nabla}$ has also non-vanishing
      coefficients: $\widetilde{\Gamma}^h_{jk}=C^h_{jk}$.
\item From the table we have $\widetilde{\Gamma}^h_{jk}=2
      \widehat{\Gamma}^h_{jk}=-C^h_{jk}$ and
      $\;\overcirc{\,\Gamma}^{h}_{jk}=-\gamma^h_{jk}=\frac{1}{2}(C^h_{kj}+C^j_{hk}+C^k_{hj})$. This means that
      the dual connection coefficients and the symmetric connection coefficients coincide (up to a constant) and are
      both equal to the global structure
      coefficients (up to a constant). On the other hand, the Levi-Civita connection
      coefficients coincide with the contortion coefficients (up to a
      sign). This shows again that everything in AP-geometry is expressible in terms of
      the global structure coefficients. Also, all surviving connections
      in the space may be represented by only one of them, say the
      Levi-Civita connection.
\item A quick look at the third column of the above table may deceive and lead to erroneous conclusions: the symmetric connection is skew-symmetric and the
      Levi-Civita connection is non-symmetric. This is by no means true. The formulation of the notion of symmetry of connections using indices
      is not applied any more in this context. In fact, a linear connection is symmetric if and only if it coincides with its dual connection,
      and this is the case for both the symmetric and Levi-Civita connections.
      Another example: although the symmetric and dual connections coincide (up to a constant) in the parallelization basis, the symmetric
      connection has no torsion while the dual connection has a surviving torsion. This is, once more, due to the fact that the torsion
      expression has a bracket term which does not depend on the connection.
\item The torsion and contortion tensors of type $(0,3)$ are present in the natural basis while they are not in the parallelization
      basis. This is because the metric matrix is the identity matrix $(\delta_{jk})$. Consequently, mesh indices can not be
      raised or lowered using the metric $g_{jk}$.
\item In local coordinates the structure coefficients vanish:
      $[\partial_{\mu},\,\partial_{\nu}]=0$ (while in the
      parallelization basis the global structure coefficients are alive:
      $[\;\undersym{X}{j},\;\undersym{X}{k}]=C^h_{jk}\; \undersym{X}{h}$). For
      this reason the structure coefficients, in local coordinates, have no
      effect and the second column of the above table give thus the usual expressions we are  accustomed to
      \cite{local}. As an example, as the connection coefficients depend on coordinate systems, the canonical
      connection coefficients do not vanish in the natural basis (while they vanish in the parallelization basis).
\item For physical applications, especially in general relativity and gravitation, one can assign a signature to the positive definite metric $g$ defined by (\ref{metric}). This can be achieved, for $n=4$, by writing \,$ g=\eta_{ij}\, \undersym{\Omega}{i}\otimes\undersym{\Omega}{j},$
    where $\eta_{ij}=0 \text{ for } i\neq j, \,\, \eta_{ij}=-1 \text{ for } i=j=0, \,\, \eta_{ij}=+1 \text{ for } i=j=1, 2, 3.$
    The metric $g$ is thus nomore positive definite but rather nondegenerate.
\end{itemize}

%%%%%%%%%%%%%%%%%%%%%%%%%%%%%%%%%%%%%%%%%%%%%%%%%%%%%%%%%%%%%%%Ref%%%%%%%%%%%%%%%%%%%%%%%%%%%%%%%%%%%%%%%%%%%%%%%%%%%%%%%%%%%%%%%

\bibliographystyle{plain}

\end{document}